%% file: petri_closures.tex
\documentclass[a4paper,USenglish]{lipics-v2016-modified}

\usepackage{style/allLIPICS}

\input{petri_closures_style}

\title
{
    On the Upward/Downward Closures of Petri Nets%
    \footnote%
    {%
        The conference version of this paper has been published in the proceedings of the 42nd International Symposium on Mathematical Foundations of Computer Science, MFCS 2017 \cite{AMMS}.
    }%
}

\author[1]{Mohamed Faouzi Atig}
\author[2]
{%
    Roland Meyer%
    \footnote%
    {%
        A part of this work was carried out when the author was at Aalto University.
    }%
}%
\author[3]{Sebastian Muskalla}
\author[4]{Prakash Saivasan}

\affil[1]
{
    Uppsala University, Sweden\\
    \texttt{mohamed\_faouzi.atig@it.uu.se}
}
\affil[2]
{
    TU Braunschweig, Germany\\
    \texttt{roland.meyer@tu-braunschweig.de}
}
\affil[3]
{
    TU Braunschweig, Germany\\
    \texttt{s.muskalla@tu-braunschweig.de}
}
\affil[4]
{
    TU Braunschweig, Germany\\
    \texttt{p.saivasan@tu-braunschweig.de}
}

\authorrunning{M. F. Atig, R. Meyer, S. Muskalla, and P. Saivasan}

\Copyright
{
    Mohamed Faouzi Atig,
    Roland Meyer,
    Sebastian Muskalla,
    and Prakash Saivasan
}

\subjclass{F.1.1 Models of Computation}

\keywords{Petri nets, BPP nets, downward closure, upward closure}

\begin{document}

\maketitle

\input{content/abstract}

\input{content/intro}

\input{content/table}
\input{content/relwork}

\input{content/pre}

\input{content/upward_pn}
\input{content/upward_bpp}

\input{content/downward_pn}

\input{content/downward_bpp}

\input{content/sre_downward_pn}
\input{content/sre_downward_bpp}

\input{content/sre_upward}

\input{content/being_ucdc}

\input{content/conclusion}

\newpage
\bibliographystyle{plainurl}
\bibliography{literatur}
 
\end{document}

%% file: petri_closures_style.tex
\newcommand{\fromto}[2]{[#1..#2]}
\renewcommand{\oneto}[1]{\fromto{1}{#1}}
\renewcommand{\zeroto}[1]{\fromto{0}{#1}}

\newcommand{\qinit}{q_{\mathit{init}}}
\newcommand{\langk}[2]{\calL_{#1} \left(#2\right)}
\newcommand{\Prj}[2]{\pi_{#2}(#1)}
\newcommand{\poly}{\text{poly}}

\newcommand{\sreu}{$\mathsf{SREU}$}
\newcommand{\sred}{$\mathsf{SRED}$}
\newcommand{\suppn}{$\mathsf{SUPPN}$}
\newcommand{\buc}{$\mathsf{BUC}$}
\newcommand{\bdc}{$\mathsf{BDC}$}

\newcommand{\Pred} [1]{ {}^{\bullet}#1}
\newcommand{\Succ} [1]{ #1^{\bullet}}
\newcommand{\tokencount}[1]{\mathit{tc}(#1)}

\newcommand{\fire}[1]{\lbrack #1 \rangle}
\newcommand {\TEnable}[3]
{{#1 \fire{#2} #3}}
\newcommand {\TJustEnable}[2]
{{#1 \fire{#2}}}
\newcommand{\move}[1] {\fire{#1}}

\newcommand{\rprun}{\mathit{run}}
\newcommand{\rptemp}{\mathit{temp}}
\newcommand{\rpstop}{\mathit{stop}}
\newcommand{\rta}{t_c}
\newcommand{\rtb}{t_b}
\newcommand{\rthelp}{t_a}

\newcommand{\firei}[2]{\fire{#1}_{#2}}
\newcommand {\Paths} {\mathit{Paths}}
\newcommand {\Words} {\mathit{Words}} 
\newcommand {\Basis} {\mathit{Basis}}
\newcommand {\SPath} {\mathit{SPath}}

\newcommand{\Inc} {\texttt{Inc}}
\newcommand{\Test}{\texttt{Test}}
\newcommand{\action}{\texttt{Action}}

\newcommand{\Acker} {\mathit{Acker}}
\newcommand{\pin} {\text{in}}
\newcommand{\pout} {\text{out}}
\newcommand{\pcopy} {\text{copy}}
\newcommand{\pstart} {\text{start}}
\newcommand{\pstop} {\text{stop}}
\newcommand{\pswap} {\text{swap}}
\newcommand{\ptmp} {\text{temp}}
\newcommand{\tcopy}{t_\mathit{copy}}
\newcommand{\tstart}{t_\mathit{start}}
\newcommand{\tstop}{t_\mathit{stop}}
\newcommand{\tswap}{t_\mathit{swap}}
\newcommand{\trestart}{t_\mathit{restart}}
\newcommand{\tin}{t_\mathit{in}}
\newcommand{\ttmp}{t_\mathit{temp}}

\newcommand{\causal}{\trianglelefteq}
\newcommand{\predec}[1]{\lfloor#1\rfloor}

\newcommand{\sre}{\mathit{sre}}
\newcommand{\product}{\mathit{p}}
\newcommand{\linof}[1]{\mathit{lin}(#1)}

\newcommand{\rightproduct}{\rtimes}

\newcommand{\deq}[1]{\leq^{#1}}
\newcommand{\TC}{\text{TC}}
\newcommand{\TE}{\text{TE}}

\newcommand{\Traces} [1]{\calT \left(#1\right) }

%% file: content/abstract.tex
\begin{abstract}
    We study the size and the complexity of computing finite state automata (FSA) representing and approximating the downward and the upward closure of Petri net languages with coverability as the acceptance condition. 
    We show how to construct an FSA recognizing the upward closure of a Petri net language in doubly-exponential time, and therefore the size is at most doubly exponential.
    For downward closures, we prove that the size of the minimal automata can be non-primitive recursive. 
    In the case of BPP nets, a well-known subclass of Petri nets, we show that an FSA accepting the downward/upward closure can be constructed in exponential time.  
    Furthermore, we consider the problem of checking whether a simple regular language is included in the downward/upward closure of a Petri net/BPP net language. 
    We show that this problem is $\EXPSPACE$-complete (\resp $\NPTIME$-complete) in the case of Petri nets (\resp BPP nets).
    Finally, we show that it is decidable whether a Petri net language is upward/downward closed.
    To this end, we prove that one can decide whether a given regular language is a subset of a Petri net coverability language.
\end{abstract}

%% file: content/intro.tex
\section{Introduction}

Petri nets are a popular model of concurrent systems~\cite{EN94}. 
Petri net languages (with different acceptance conditions) have been extensively studied during the last years, including deciding their emptiness
(which can be reduced to reachability) \cite{Mayr1984,Kosaraju82,Lambert1992,DBLP:conf/popl/Leroux11},
their regularity \cite{DBLP:journals/jcss/ValkV81,Demri2013},
their context-freeness \cite{DBLP:journals/tcs/Schwer92a,DBLP:conf/lics/LerouxPS13},
and many other decision problems (\eg \cite{HMW10,DBLP:conf/csl/AbdullaDB07,DBLP:journals/entcs/FinkelGRB05}). 
In this paper, we consider the class of Petri net languages with coverability as the acceptance condition (\ie the set of sequences of transition labels occurring in a computation reaching a marking greater than or equal to a given final marking).

We address the problem of computing the \emph{downward} and the \emph{upward} closure of Petri net languages. 
The downward closure of a language $\calL$, denoted by $\dc{\calL}$, is the set of all subwords, all words that can be obtained from words in $\calL$ by deleting letters.
The upward closure of $\calL$, denoted by $\uc{\calL}$, is the set of all superwords, all words that can be obtained from words in $\calL$ by inserting letters.   
It is well-known that, for any language, the downward and upward closure are regular and can be described by a \emph{simple regular expression (SRE)}.
However, such an expression is in general not computable, \eg for example, it is not possible to compute the downward closure of languages recognized by lossy channel systems \cite{Mayr03}.

In this paper, we first consider the problem of constructing a finite state automaton~(FSA) accepting the upward/downward closure of a Petri net language.
We give an algorithm that computes an FSA of doubly-exponential size for the upward closure in doubly-exponential time.
This is done by  showing that every minimal word results from a computation of length at most doubly exponential in the size of the input.
Our algorithm is also optimal since we present a  family of Petri net languages  for which the minimal finite state automata representing their upward closure are of doubly-exponential size.

Our second contribution is a family of Petri net languages for which the size of the minimal finite state automata representing the downward closure is non-primitive recursive.
To prove this, we resort to a construction due to Mayr and Meyer \cite{MayrMeyer1981}.
It gives a family of Petri nets whose language, albeit finite, is astronomically large:
It contains Ackermann many words.
The downward closure of Petri net languages has been shown to be effectively computable~\cite{HMW10}.
The algorithm is based on the Karp-Miller tree~\cite{km69}, which has non-primitive recursive~complexity. 

Furthermore, we consider the SRE inclusion problem which asks whether the language of a simple regular expression is included in the downward/upward closure of a Petri net language.
The idea behind SRE inclusion is to stratify the problem of computing the downward/upward closure:
Rather than having an algorithm computing all information about the language, we imagine to have an oracle (\eg an enumeration) making proposals for SREs that could be included in the downward/upward closure.
The task of the algorithm is merely to check whether a proposed inclusion holds.
We show that this problem is $\EXPSPACE$-complete in both cases.
In the case of upward closures, we prove that SRE inclusion boils down to checking whether the set of minimal words of the given SRE is included in the upward closure.
In the case of downward closures, we reduce the problem to the simultaneous unboundedness problem for Petri nets, which is $\EXPSPACE$-complete~\cite{Demri2013}.

We also study the problem of checking whether a Petri net language actually is 
upward or downward closed. 
This is interesting as it means that an automaton for the closure, which we can compute with the aforementioned methods, is a precise representation of the system's behavior.  
We show that the problem of being upward/downward closed is decidable for Petri nets. 
The result is a consequence of a more general decidability that we believe is of independent interest.  
We show that checking whether a regular language is included in a Petri net language (with coverability as the acceptance condition) is decidable.
Here, we rely on a decision procedure for trace inclusion due to Esparza et al.~\cite{Esparza99}.

Finally, we consider \emph{BPP\,\footnote{BPP stands for \emph{basic parallel processes}, a notion from process algebra.} nets} \cite{Esparza19972}, 
a subclass of Petri nets defined by a syntactic restriction: 
Every transition is allowed to consume at most one token in total.
We show that we can compute finite state automata accepting the upward and the downward closure of a BPP net language in exponential time.
The size of the FSA is also exponential. 
Our algorithms are optimal as we present a family of BPP net languages for which the minimal FSA representing their upward/downward closure have exponential size.
Furthermore, we consider the SRE inclusion problem. 
We show that, in the case of BPP nets, it is $\NPTIME$-complete for both, inclusion in the upward and in the downward closure.
To prove the upper bound, we reduce to the satisfiability problem for existential Presburger arithmetic (which is known to be $\NPTIME$-complete \cite{Scarpellini1983}). 
The hardness is by a reduction from $\mathsf{SAT}$ to the emptiness of BPP net languages, which in turn reduces to SRE inclusion.

%% file: content/table.tex
\vspace*{1em}

\noindent
The following table summarizes our results.

\setlength\extrarowheight{4pt}

\begin{center}
    \begin{tabu}{l||X[0.6,c]|X[0.4,c]}
        & Petri nets & BPP nets
        \\
        \hline
        \hline
        Computing the upward closure & Doubly exponential$^*$ & Exponential$^*$ \\
        \hline
        Computing the downward closure & Non-primitive recursive$^*$ & Exponential$^*$ \\
        \hline
        SRE in downward closure & $\EXPSPACE$-complete & $\NPTIME$-complete  \\
        \hline
        SRE in upward closure & $\EXPSPACE$-complete & $\NPTIME$-complete \\
        \hline
        Being downward/upward closed &\multicolumn{2}{c}{Decidable} \\
        \hline
        Containing regular language&\multicolumn{2}{c}{Decidable} \\
    \end{tabu}
\end{center}

    ${}^* \colon$ Time for the construction and size of the resulting FSA, optimal.

\setlength\extrarowheight{0pt}

%% file: content/relwork.tex
\paragraph*{Related Work}

Several constructions have been proposed in the literature to compute finite state automata recognizing the downward/upward closure.
In the case of Petri net languages (with various acceptance conditions including reachability), it has been shown that the downward closure is effectively computable \cite{HMW10}. 
With the results in this paper, the computation and the state complexity have to be non-primitive recursive.
For the languages generated by context-free grammars, effective computability of the downward closure is due to \cite{vanLeeuwen1978,DBLP:conf/mcu/GruberHK07,Courcelle91,BLS15}. 
For the languages recognized by one-counter automata,
a strict subclass of the context-free languages, it has been shown how to compute in polynomial time a finite state automaton accepting the downward/upward closure of the language \cite{DBLP:journals/corr/AtigCHKSZ16}. 
The effective computability of the downward closure has also been shown for stacked counter automata~\cite{DBLP:conf/stacs/Zetzsche15}.  
In \cite{Zetzsche15}, Zetzsche provides a characterization for a class of languages to have an effectively computable downward closure. 
It has been used to prove the effective computability of downward closures of higher-order pushdown automata and higher-order recursion schemes \cite{DBLP:conf/popl/HagueKO16,ClementePSW16}. 
The downward closure of the languages of lossy channel systems is not computable~\cite{Mayr03}. 

The computability results discussed above have been used to prove the decidability of verification problems and to develop approximation-based program analysis methods (see \eg \cite{DBLP:conf/concur/AtigBT08,DBLP:journals/corr/abs-1111-1011,DBLP:conf/fsttcs/AtigBKS14,DBLP:conf/concur/TorreMW15,LanguageRefinement, Zetzsche16}). 
Throughout the paper, we will give hints to applications in verification.

%% file: content/pre.tex

\section{Preliminaries}
\label{Section:Preliminaries}

In this section, we fix some basic definitions and notations that will be used throughout the paper. 
For every $i, j \in \N$, we use $\fromto{i}{j}$ to denote the set $\Set{k \in \N }{ i \leq k \leq j }$.

\paragraph*{Languages and Closures}

Let $\Sigma$ be a finite alphabet.
We use $\Sigma_\varepsilon$ to denote $\Sigma \cup \set{\varepsilon}$.
The length of a word $u$ over $\Sigma$ is denoted by $\card{u}$, where $\card{\varepsilon}=0$.
Let $k \in \N$ be a natural number, we use $\Sigma^k$ (\resp $\Sigma^{\leq k}$) to denote the set of all words of length equal (\resp smaller or equal) to $k$.
A language $\calL$ over $\Sigma$ is a (possibly infinite) set of finite words over $\Sigma$.

Let $\Gamma$ be a subset of $\Sigma$.
Given a word $u \in \Sigma^*$, we denote by $\Prj{u}{\Gamma}$ the projection of $u$ over $\Gamma$, \ie the word obtained from $u$ by erasing all the letters that are not in $\Gamma$.

The \emph{Parikh image} of a word~\cite{Parikh1966} counts the number of occurrences of all letters while forgetting about their positioning.
Formally, the function $\Parikh : \Sigma^* \mapsto \N^{\Sigma}$ takes a word $w\in\Sigma^*$ and gives the function $\Parikh(w):\Sigma\rightarrow\N$ defined by $(\Parikh(w))(a) =\card{\Prj{w}{\{a\}}}$ for all $a\in \Sigma$.

The  \emph{subword relation} $\subword\ \, \subseteq \Sigma^*\!\times\Sigma^*$ \cite{Higman52} between words is defined as follows: 
A word $u = a_1 \ldots a_n$ is a subword of $v$, denoted $u \subword v$, if $u$ can be obtained by deleting letters from $v$ or, equivalently, if $v = v_0 a_1 v_1 \ldots a_n v_{n}$ for some \mbox{$v_0, \ldots, v_{n} \in \Sigma^*$.}

Let $\calL$ be a language over $\Sigma$.
The \emph{upward closure} of $\calL$ consists of all words that have a subword in the language,
$\uc{\calL} = \Set{v\in\Sigma^*}{\exists u \in \calL \colon u \subword v}$.
The \emph{downward closure} of $\calL$ contains all words that are dominated by a word in the language, 
$\dc{\calL} = \Set{ u\in\Sigma^* }{\exists v \in \calL \colon u \subword v}$.
Higman showed that the subword relation is a well-quasi ordering~\cite{Higman52}, which means that every set of words $\calL \subseteq \Sigma^*$ has a finite \emph{basis}, a finite set of \emph{minimal elements} $v\in \calL$ such that $\nexists u\in \calL: u \neq v, u\subword v$.
With finite bases, $\uc{\calL}$ and $\dc{\calL}$ are guaranteed to be regular for every language $\calL\subseteq \Sigma^*$~\cite{Haines1969}.
Indeed, they can be expressed using  the subclass of simple regular languages defined by so-called \emph{simple regular expressions}~\cite{Bouajjani2004}. 

These SREs are choices among \emph{products} $\product$ that interleave single letters $a$ or $(a+\varepsilon)$ with iterations over letters from subsets $\Gamma\subseteq \Sigma$ of the alphabet:
\[
    \sre::=\product\bnf\sre+\sre
    \qquad
    \product::=a \bnf (a+\varepsilon)\bnf\Gamma^*\bnf\product.\product\
    \ .
\]
Note that this is an extension of the classical definition that we introduce so that we can also represent upward closures.
The syntactic size of an SRE $\sre$ is denoted by $\card{\sre}$ and defined as expected, every piece of syntax contributes to it.

\paragraph*{Finite State Automata}

A \emph{finite state automaton (FSA)} $A$ is a tuple $(\Sigma,Q,\rightarrow,\qinit, Q_f)$ where
$Q$ is a finite non-empty set of states,
$\Sigma$ is the finite input alphabet,
$\rightarrow\ \subseteq Q \times \Sigma_\varepsilon \times Q$ is the non-deterministic transition relation,
$\qinit \in Q$ is the initial state, and
$Q_f\subseteq Q$ is the set of final states. 
We represent a transition $(q,a,q') \in \, \rightarrow$ by $q \tow{a} q'$ and generalize the relation to words in the expected way.
The language of finite words accepted by $A$ is denoted by $\lang{A}$.
The size of $A$, denoted $\card{A}$, is defined by $\card{Q}+ \card{\Sigma}$.
An FSA is \emph{minimal} for its language $\lang{A}$ if there is no FSA $B$ with $\lang{A} = \lang{B}$ with a strictly smaller number of states.

\paragraph*{Petri Nets}

A \emph{(labeled) Petri net}  is a tuple \mbox{$N = (\Sigma,P,T,F,\lambda)$} \cite{Reisig1985}.
Here, $\Sigma$ is a finite alphabet, $P$ a finite set of \emph{places}, $T$ a finite set of \emph{transitions} with $P \cap T = \emptyset$, \mbox{$F: (P \times T) \cup (T \times P) \rightarrow \N$} a \emph{flow function}, and $\lambda: T \mapsto\Sigma_\varepsilon$ a labeling function.
When convenient, we will assume that the places are ordered, $P = \oneto{\ell}$ for some $\ell \in \N$.
For a place or transition $x \in P \cup T$, we define the \emph{preset} to consist of the elements that have an arc to $x$, $\Pred{x}~=~\Set{ y \in P \cup T}{ F(y,x) > 0 }$.
The \emph{postset} is defined similarly,
$\Succ{x}~=~\Set{ y \in P \cup T }{ F(x,y) > 0 }$.

To define the semantics of Petri nets, we use \emph{markings} $M: P \rightarrow \N$
that assign to each place a number of \emph{tokens}.
A marking $M$ \emph{enables} a transition $t$, denoted $\TJustEnable {M} {t}$, if \mbox{$M(p) \geq F(p,t)$} for all $p \in P$.
A transition $t$ that is enabled may be \emph{fired}, leading to the new marking $M'$ defined by $M'(p) = M(p) - F(p,t) + F(t,p)$ for all $p \in P$, \ie $t$~consumes $F(p,t)$ tokens and produces $F(t,p)$ tokens in $p$.
We write the firing relation as $\TEnable {M} {t} {M'}$.
A \emph{computation}
\(
    \pi = M_0 \move{t_1} M_1 \cdots \move{t_m} M_m
\)
consists of markings and transitions.
We extend the firing relation to transition sequences $\sigma\in T^*$ in the straightforward manner and also write $\pi = M_0\move{\sigma} M_m$. 
A marking $M$ is \emph{reachable} from an initial marking $M_0$ if $\TEnable {M_0}{\sigma} {M}$ for some $\sigma\in T^*$.
A marking $M$ covers another marking $M_f$, denoted $M \geq M_f$, if \mbox{$M(p) \geq M_f(p)$} for all $p \in P$.
A marking $M_f$ is \emph{coverable} from $M_0$ if there is a marking $M$ reachable from $M_0$ that covers $M_f$, $\TEnable{M_0}{\sigma}{M}\geq M_f$ for some $\sigma\in T^*$.

A \emph{Petri net instance} $(N,M_0,M_f)$ consists of a Petri net $N$ together with an initial marking $M_0$ and final marking $M_f$ for $N$.
Given a Petri net instance $(N,M_0,M_f)$, the associated \emph{covering language} is
\[
    \lang{N,M_0,M_f} = \Set{ \lambda(\sigma)}{\sigma \in T^*,\  \TEnable{M_0}{\sigma}{M}\geq M_f }
    \ ,
\]
where  the labeling function $\lambda$ is extended to sequences of transitions in the straightforward manner. Given a natural number $k \in \N$, we define
\[
    \langk{k}{N,M_0,M_f} = \Set{ \lambda(\sigma)}{\sigma \in T^{\leq k},\ \TEnable{M_0}{\sigma}{M}\geq M_f }
\]
to be the set of words accepted by computations of length at most $k$.

Let ${\it max}(F)$ denote the  maximum of the range of $F$.
The size of the Petri net $N$ is
\[
    |N|= |\Sigma|+ |P|\cdot |T| \cdot (1+ \lceil{\it log}_2(1+{\it max}(F))\rceil )
    \ ,
\]
Similarly, the size of a marking $M$ is
\[
    |M|=  |P| \cdot ( 1+ \lceil{\it log}_2(1+{\it max}(M))\rceil )
    \ ,
\]
where ${\it max}(M)$ denotes the maximum of the range of $M$.
The size of a Petri net instance $(N,M_0,M_f)$ is $\card{(N,M_0,M_f)} = \card{N} + \card{M_0} + \card{M_f}$.
This means we consider the the binary encoding of numbers occurring in markings and the flow function.
In contrast, we define the \emph{token count} $\tokencount{M} = \Sigma_{p\in P}M(p)$ of a marking $M$ to be the sum of all tokens assigned by $M$, \ie the size of the unary encoding of $M$.

A Petri net $N$ is said to be a \emph{BPP net} (or \emph{communication-free}) if every transition consumes at most one token from one place (\ie $\Sigma_{p\in P}F(p, t)\leq 1$ for every $t \in T$).

%% file: content/upward_pn.tex

\section{Upward Closures}
\label{Section:UpwardClosure}

We consider the problem of constructing a finite state automaton accepting  the upward closure of a Petri net and a BPP net language, respectively.
The upward closure offers an over-approximation of the system behavior that is useful for verification purposes~\cite{LanguageRefinement}.

\begin{compproblem}
    \problemtitle{Computing the upward closure}
    \probleminput{A Petri net instance $(N,M_0,M_f$).}
    \problemquestion{An FSA $A$ with $\lang{A} = \uc{\lang{N,M_0,M_f}}$.}	
\end{compproblem}

\subsection{Petri Nets}

We prove a doubly-exponential upper bound on the size of the finite state automaton representing the upward closure of a  Petri net language. 
Then, we present a family of Petri net languages for which the minimal finite state automata representing their upward closure have a size doubly exponential in the size of the input.

\paragraph*{Upper Bound}

Fix the Petri net instance $(N,M_0,M_f)$ of interest and let $n$ be its size.

\begin{theorem}
\label{Theorem:UCCompPN}
    One can construct an FSA of size  $\bigO{2^{2^{\poly(n)}}}$ for $\uc{\lang{N,M_0,M_f}}$.
\end{theorem}
The remainder of the section is devoted to proving the theorem. 
We will show that every minimal word results from a computation of length at most $O(2^{2^{\poly(n)}})$. 
Let us call such computations the minimal ones. 
Let $k$ be a bound on the length of the minimal computations. This means the language $\langk{k}{N,M_0,M_f}$ contains all minimal words of $\lang{N,M_0,M_f}$. 
Furthermore, $\langk{k}{N,M_0,M_f} \subseteq \lang{N,M_0,M_f}$ and therefore the equality $\uc{\langk{k}{N,M_0,M_f}}=\uc{\lang{N,M_0,M_f}}$ holds. 
Now we can use the following lemma to construct a finite automaton whose size is $\bigO{2^{2^{\poly(\card{n})}}}$ and that accepts $\langk{k}{N,M_0,M_f}$. 
Without an increase in size, this automaton can be modified to accept $\uc{\langk{k}{N,M_0,M_f}}$:
Add for each state $q$ and for each symbol $a \in \Sigma$ a loop $q \tow{a} q$.

\begin{lemma}
\label{kboundedlanguages}
    For every $k \in \N$, one can construct an FSA of size  $\bigO{(k+2)^{\poly({n})}}$ that  accepts $\langk{k}{N,M_0,M_f}$.
\end{lemma}

\begin{proof}
    If there is a word $w \in \langk{k}{N,M_0,M_f}$, then there is a run of the form $M_0\move{\sigma}M'$ with $M' \geq M_f$ and $|\sigma| \leq k$.
    Any place $p$  can have at most $M_0(p)+ k \cdot 2^n$ tokens in $M'$.
    Note that $M_0(p) \leq 2^n$.
    With this observation, we construct the required finite state automaton $A = (\Sigma,Q,\to,\qinit,Q_f)$ as follows.
    
    The set of states is $Q = (P \to \zeroto{(k+1) \cdot 2^n}) \times \zeroto{k} $.
    The first component stores the token count for each place $p \in \oneto{\ell}$ (\ie a marking), the second component counts the number of steps that have been executed so far.
    For each transition $t \in T$ of the Petri net and each state $(M,i)$ with $M(p) \geq F(p,t)$ for all $p$ and $i < k$, there is a transition from $(M,i)$ to $(M',i+1)$ in $\to$, where $M \move{t} M'$.
    It is labeled by $\lambda(t)$.
    The initial state is $\qinit = (M_0,0)$, and a state $(M',i)$ is final if $M'$ covers $M_f$.
    By the construction of the automaton, it is clear that $(M_0,0) \tow{w} (M',j)$ with $(M',j)$ final iff there is a $\sigma \in T^{\leq k}$ such that $M_0\move{\sigma} M'$ with $M' \geq M_f$.
    Hence we have $\lang{A} = \langk{k}{N,M_0,M_f}$.
    We estimate the size of $Q$ by
    \begin{align*}
        |Q|&= |(P \to \zeroto{(k+1) \cdot 2^n})| \cdot |\zeroto{k}|\\
        &= ((k+1) \cdot 2^n)^\ell \cdot (k+1)
        \leq  ((k+1) \cdot 2^n)^n \cdot (k+1)
        = (k+1)^{n+1} \cdot 2^{n^2}\\
        &\leq (k+2)^{n+1} \cdot (k+2)^{n^2}
        \leq (k+2)^{n+1 + n^2} \in \bigO{(k+2)^{\poly(n)}}
        \ .
    \end{align*}
    Since we can assume $|\Sigma| \leq n$, we have that $|A|=|Q| + |\Sigma| $ is in $\bigO{(k+2)^{\poly(n)}}$.
\end{proof}
It remains to show that every minimal word results from a computation of length at most doubly exponential in the size of the input. 
This is the following proposition.

\begin{proposition}
    \label{rackoffproof}
    For every computation $M_0\move{\sigma}M\geq M_f$, there is $M_0\move{\sigma'}M'\geq M_f$ with $\lambda(\sigma') \subword \lambda(\sigma)$ and $\card{\sigma'} \leq 2^{2^{cn\log n}}$, where $c$ is a constant.
\end{proposition}
Our proof is an adaptation of Rackoff's technique to show that coverability can be solved in $\EXPSPACE$~\cite{Rackoff78}. 
Rackoff derives a bound (similar to ours) on the length of the shortest computations that cover a given marking.
Rackoff's proof has been generalized to different settings, \eg to BVAS in~\cite{DEMRI201323}. 
Lemma 5.3 in~\cite{Leroux2013} claims that Rackoff's original proof already implies Proposition~\ref{rackoffproof}.
This is not true as shown by the following example.

\begin{example}
    Consider the Petri net $N_{\mathit{ce}} = (\set{a,b,c},\set{ \rprun, \rptemp, \rpstop  }, \set{ \rthelp, \rtb, \rta  }, F, \lambda)$, where the flow relation is given by Figure~\ref{Figure:Rackoff} and we have $\lambda(\rthelp) = a, \lambda(\rtb) = b$, and $\lambda (\rta) = c$.
    \begin{figure}[!ht]
        \centering
        \input{tikz/rackoff_counterexample.tikz}
        \caption{\!\!\textbf{\sffamily .} The flow relation of the Petri net $N_{\mathit{ce}}$.}
        \label{Figure:Rackoff}
    \end{figure}
    Consider the initial marking $M_0 = (1,0,0)$ with one token on $\rprun$ and no token elsewhere, and the final marking $M_f = (0,0,1)$ that requires one token on $\rpstop$.
    We have $\lang{N,M_0,M_f} = a^+b \cup a^*c$ and thus $\uc{\lang{N,M_0,M_f}} \, =  \Sigma^*a\Sigma^*b\Sigma^* \cup \Sigma^*c\Sigma^*$.
    
    We may compute Rackoff's bound~\cite{Rackoff78}, and obtain that if there is a computation covering $M_f$ from $M_0$, then there is one consisting of at most one transition.
    Indeed, the computation $M_0 \move{\rta} M_f$ is covering.
    
    However, computations whose length is within Rackoff's upper bound do not necessarily generate all minimal words:
    We have that $ab$ is a minimal word in $\uc{\lang{N,M_0,M_f}}$, but the shortest covering computation that generates $ab$ is $M_0 \move{\rthelp} (1,1,0) \move{\rtb} M_f$, consisting of $3$ markings and $2$ transitions.
\end{example}
To handle labeled Petri nets, Rackoff's proof needs two amendments. 
First, it is not sufficient to consider the shortest covering computations. 
Instead, we have to consider computations long enough to generate all minimal words. 
Second, Rackoff's proof splits a firing sequence into two parts and replaces the second part by a shorter one. In our case, 
we need that the shorter word is a subword of the original one.

We now elaborate on Rackoff's proof strategy and give the required definitions, then we explain in more detail our adaptation, and finally give the technical details.

We assume that the places are ordered, \ie $P = \oneto{\ell}$. 
Rackoff's idea is to relax the definition of the firing relation and allow for negative token counts on the last $i+1$ to $\ell$ places.
With a recurrence over the number of places, he then obtains a bound on the length of the computations that keep the first $i$ places positive. 

Formally, an \emph{pseudo-marking} of $N$ is a function $M: P \rightarrow \Z$.
For $i \in \oneto{\ell}$, a pseudo-marking marking $M$ \emph{$i$-enables} a transition $t \in T$ if $M(j) \geq F(j,t)$ for all $j \in \oneto{i}$. 
Firing $t$ yields a new pseudo-marking $M'$, denoted $M \firei{t}{i} M'$, with $M'(p) = M(p) - F(p,t) + F(t,p)$ for all $p \in P$.  
A computation
\(
    \pi=M_0 \firei{t_1}{i} M_1\ldots \firei{t_m}{i} M_m
\)
is \emph{$i$-non-negative} if for each marking $M_k$ with $k \in \zeroto{m}$ and each place $j \in \oneto{i}$, we have $M_k(j) \geq 0$.
We assume a fixed marking $M_f$ to be covered.
The computation $\pi$ is \emph{$i$-covering} (\wrt $M_f$) if $M_m(j) \geq M_f(j)$ for all $j \in \oneto{i}$.
Given two computations $\pi_1 = M_0 \firei{t_1}{i} \cdots \firei{t_k}{i} M_k$ and $\pi_2 = M'_0 \firei{t'_1}{i} \cdots \firei{t'_s}{i} M'_{s}$ such that $M_k(j) = M'_0(j)$ for all $j \in \oneto{i}$, we define their \emph{$i$-concatenation} $\pi_1\cdot_{i}\pi_2$ to be the computation
$M_0 \firei{t_1}{i} \cdots \firei{t_k}{i} M_k \firei{t'_1}{i} M''_{k+1} \cdots \firei{t'_s}{i} M''_{k+s}$. 
   
Rackoff's result provides a bound on the length of the shortest $i$-covering computations. 
Since we have to generate all minimal words, we will specify precisely which computations to consider (not only the shortest ones). 
Moreover, Rackoff's bound holds independent of the initial marking. 
This is needed, because the proof of the main lemma splits a firing sequence into two parts and then considers the starting marking of the second part as the new initial marking.
The sets we define in the following will depend on some unrestricted initial marking $M$, but we then quantify over all possible markings to get rid of the dependency.

Let $\Paths(M,i)$ be the set of all $i$-non-negative and $i$-bounded computations from $M$,
\[
    \Paths(M,i) = \Set{ \sigma\in T^* }{ \pi = M \move{\sigma}_i M',\  \pi \text{ is $i$-non-negative and $i$-covering} }
    \ .
\]
Let
\(
    \Words(M,i) = \Set{ \lambda(\sigma) }{ \sigma \in \Paths(M,i) }
\)
be the corresponding set of words, and let
\(
    \Basis(M,i) = \Set{ w \in \Words(M,i) }{ w\text{ is $\subword$-minimal}}
\)
be its minimal elements.
The central definitions is $\SPath(M,i)$, the set of shortest paths yielding the minimal words in $\Basis(M,i)$,
\[
    \SPath(M,i)
    =
    \Set
    {
        \sigma \in \Paths(M,i)
    }
    {
        \begin{array}{ll}
            \lambda(\sigma) \in \Basis(M,i),&
            \\
            \nexists\ \sigma' \in \Paths(M,i) \colon
            &\card{\sigma'} < \card{\sigma}, \ \lambda(\sigma') = \lambda(\sigma)
        \end{array}
    }
    \ . 
\]
Define $m(M,i) = \max \Set{ \card{\sigma} + 1}{ \sigma \in \SPath(M,i) }$ to be the length ($+1$) of the longest path in $\SPath(M,i)$, or \mbox{$m(M,i) = 0$} if $\SPath(M,i)$ is empty.
Note that $\Basis(M,i)$ is finite and therefore only finitely many different lengths occur for sequences in $\SPath$, \ie $m(M,i)$ is well-defined.
To remove the dependency on $M$, define
\[
    f(i) = \max \Set{ m(M,i)} {M \colon P \to \Z}
\]
to be the maximal length of an $i$-covering computation, where the maximum is taken over all unrestricted initial markings. 
The well-definedness of $f(i)$ is not clear yet and will be a consequence of the next lemma.
A bound on $f(\ell)$ will give us a bound on the maximum length of a computation accepting a minimal word from $\lang{N,M_0,M_f}$. 
To derive the bound, we prove that $f(i+1) \leq (2^n f(i))^{i+1} + f(i)$ using Rackoff's famous case distinction~\cite{Rackoff78}.

\begin{lemma}
\label{rackoffbound}
    $f(0) = 1$ and $f(i+1) \leq (2^n f(i))^{i+1} + f(i)$ for all $i\in \oneto{\ell-1}$.
\end{lemma}

\begin{proof}
    To see that $f(0) = 1$, note that $\varepsilon~\in~\Basis(M,0)$ for any $M \in \Z^{\ell}$, and the empty firing sequence is a $0$-covering sequence producing $\varepsilon$.
    
    For the second claim, we show that
    for any $M \in \Z^{\ell}$ and any $w \in \Words(M,i+1)$, we can find \mbox{$ \sigma \in \Paths(M,i+1)$} with $\card{\sigma} < (2^n f(i))^{i+1} + f(i)$ and $\lambda(\sigma) \subword w$. 
    Let $\sigma' \in~\Paths(M,i+1)$ be a shortest firing sequence of transitions such that $\lambda(\sigma') = w$.
    If \mbox{$\card{\sigma'} < (2^n f(i))^{i+1} + f(i)$}, we have nothing to do.
    Assume now $\card{\sigma'}~\geq~(2^n f(i))^{i+1} + f(i)$. 
    We distinguish two cases.
    
    \noindent
    \textbf{\sffamily\nth{1} Case:}
    Suppose $\sigma'$ induces the $(i+1)$-non-negative, \mbox{$(i+1)$-covering} computation $\pi'$, in which for each occurring marking $M$ and for each place \mbox{$p \in \oneto{i+1}$,} $M(p) < 2^n \cdot f(i)$ holds.
    We extract from $\pi'$ an $(i+1)$-non-negative, $(i+1)$-covering computation $\pi$ where no two markings agree on the first $(i+1)$ places:
    Whenever such a repetition occurs in $\pi'$, we delete the transitions between the repeating markings to obtain a shorter computation that is still $(i+1)$-covering.
    Iterating the deletion yields the sequence of transition $\sigma$.
    The computation $\sigma$ satisfies
    \[
    \card{\sigma} <  (2^n f(i))^{i+1}\leq (2^n f(i))^{i+1} + f(i)
    \ .
    \]
    The strict inequality holds as a computation of $h$ markings has $(h-1)$ transitions.
    Moreover, $\sigma$ is a subword of the original $\sigma'$, and hence $\lambda(\sigma) \subword \lambda(\sigma') = w$.
    
    \noindent
    \textbf{\sffamily\nth{2} Case:}
    Otherwise, $\sigma'$ is the path of an $(i+1)$-non-negative, $(i+1)$-covering computation $\pi'$, in which a marking occurs that assigns more than $2^n \cdot f(i)$ tokens to some place $p \in \oneto{i+1}$.
    Then, we can decompose $\pi'$ as follows:
    \begin{align*}
    \pi' = M\firei{\sigma_1'}{i+1} M_1 \firei{t}{i+1} M_2 \move{\sigma_2'}_{i+1}M'
    \end{align*} 
    so that $M_2$ is the first marking that assigns $2^n \cdot f(i)$ or more tokens to some place, say \wolog place $i+1$. 
    We may assume that \mbox{$\card{\sigma_1'}  < (2^nf(i))^{i+1}$}.
    Otherwise, we can replace $\sigma_1'$ by a repetition-free sequence $\sigma_1$ as in the first case, where $M_0\firei{\sigma_1}{i+1}M'_1$
    such that $M_1'$ and $M_1$ agree on the first $i+1$ places.
    
    Note that $\pi_2'=M_2\firei{\sigma_2'}{i+1}M'$ is also an $i$-non-negative, $i$-covering computation. 
    By the definition of $f(i)$, there is an $i$-non-negative, $i$-covering computation $\pi_2$ starting from $M_2$ such that the corresponding path $\sigma_2$ satisfies $\card{\sigma_2} < f(i)$ and $\lambda(\sigma_2) \subword \lambda(\sigma_2')$.
    Since the value of place $i+1$ is greater or equal $2^nf(i)$, it is easy to see that $\pi_2$ is also an $(i+1)$-non-negative, $(i+1)$-covering computation starting in $M_2$: Even if all the at most $f(i)-1$ transitions subtract $2^n$ tokens from place $i+1$, we still end up with $2^n$ tokens.            
    The concatenation $\sigma'_1 \cdot_i t \cdot_i \sigma'_2$ is then an $(i+1)$-non-negative, $(i+1)$-covering run starting in $M$ of length at most $((2^nf(i))^{i+1}-1) + 1 + (f(i)-1) < (2^nf(i))^{i+1}+ f(i)$.
\end{proof}

\begin{proof}[Proof of Proposition~\ref{rackoffproof}]
    As in~\cite{Rackoff78}, we define the function $g$ inductively by $g(0) = 2^{3n}$ and $g(i + 1) = (g(i))^{3n}$.
    It is easy to see that $g(i) = 2^{((3n)^{(i+1)})}$.
    Using Lemma~\ref{rackoffbound}, we can conclude $f(i) \leq g(i)$ for all $i \in \zeroto{\ell}$.
    Furthermore,
    \[
    f(\ell)
    \leq g(\ell)
    \leq 2^{((3n)^{(\ell+1)})}
    \leq 2^{((3n)^{n+1})}
    \leq 2^{2^{c n \log n}}
    \]
    for some suitable constant $c$.
    
    Let $M_0 \move{\sigma} M \geq M_f$ be a covering computation of the Petri net.
    By the definitions, $\sigma \in \Paths(M_0, \ell)$ and $\lambda(\sigma) \in \Words(M_0,\ell)$.
    There is a word $w \in \Basis(M_0, \ell)$ with $w \subword \lambda(\sigma)$, and $w$ has a corresponding computation $\sigma' \in \SPath(M_0,\ell)$ (\ie $\lambda(\sigma') = w$).
    By the definition of $f(\ell)$, we have $\card{\sigma'} < m(M_0, \ell) \leq f(\ell) \leq 2^{2^{c n \log n}}$.
\end{proof}

\paragraph*{Lower Bound}
We present a family of Petri net languages for which the minimal finite state automata representing the upward closure are of size doubly exponential in the size of the input. 
We rely on a construction due to Lipton \cite{Lipton} that shows how to calculate in a precise way (including zero tests) with values up to $2^{2^n}$ in Petri nets.

\begin{lemma}
\label{lipton}
    For every number $n \in N$, we can construct a Petri net
    \mbox{$N(n) = (\{a\},P,T,F,\lambda)$}
    and markings $M_0, M_f$ of size polynomial in $n$ such that
    $\lang{N(n),M_0,M_f} = \big\{ a^{2^{2^n}} \big\}\ .$
\end{lemma}

\begin{proof}
    We rely on Lipton's proof~\cite{Lipton} of $\EXPSPACE$-hardness of Petri net reachability.
    Lipton shows how a counter machines in which the counters are bounded by $2^{2^n}$ can be simulated using a Petri net of polynomial size.
    We will use the notations as in \cite{Esparza1998}.
    
    Lipton defines \emph{net programs} (called \emph{parallel programs} in~\cite{Lipton}) to encode Petri nets.
    For the purpose of proving this lemma, we will recall the syntax of net programs and also some of the subroutines as defined in~\cite{Lipton,Esparza1998}.
    
    We will use the following commands \resp the following established subroutines from~\cite{Esparza1998} in the program.
    
    \vspace*{0.2cm}
    
    \begin{tabu}{ll}
        $l: x := x - 1$
        &decrement a variable $x$
        \\
        $l: \texttt{gosub} \; s$
        & call the subroutine $s$
        \\
        $l: \texttt{Inc}_n(x)$
        &sets variable $x$ to exactly $2^{2^n}$
        \\
        $l: \texttt{Test} (x,l_{=0},l_{\neq 0})$
        & {jumps to $l_{=0}$ if $x = 0$ and to $l_{\neq0}$ if $x \neq 0$}
        \\
    \end{tabu}
    
    \vspace*{0.2cm}
        
    \noindent
    Note that all commands can be encoded using a Petri net of size polynomial in $n$.
    The fact that a test for zero can be implemented (by the subroutine $\texttt{Test} (x,l_{=0},l_{\neq 0})$) relies on the counters being bounded by $2^{2^n}$.
    It is not possible to encode zero tests for counter machines with unbounded counters using Petri nets.
    
    We assume that in the Petri net encoded by these commands, all transitions are labeled by $\varepsilon$.
    We consider an additional command $\action(a)$ to accept the input $a$, which can be encoded using a set of transitions such that exactly one of them is labeled by $a$.
    
    Consider the following net program.
    
    \begin{tabu}{l@{\ }l}
        $ l_1:$ &$\texttt{gosub} \; \Inc_n(x) $\\
        $ l_2:$ &$x := x-1 $\\
        $ l_3:$ &$\action(a)$\\
        $ l_4:$ &$\texttt{gosub} \; \Test(x,l_5,l_2) $\\
        $l_5:$ &$\texttt{Halt}$
    \end{tabu}
    
    \noindent
    In any halting computation, the program performs $\action(a)$ exactly $2^{2^{n}}$ times.
    The required Petri net $N(n)$ is the one equivalent to this net program.
\end{proof}
The upward closure $\uc{\lang{N(n),M_0,M_f}}$ is $\Set{ a^k }{k \geq 2^{2^n}}$ and needs at least $2^{2^n}$ states.

%% file: tikz/rackoff_counterexample.tikz
\begin{tikzpicture}[>=stealth',shorten >=1pt,auto,node distance=2cm]


\node[place] (run)  at (0,1) {$\rprun$};

\node[place] (temp) at (2,3) {$\rptemp$};

\node[place] (stop) at (6,1) {$\rpstop$};

\node[transition] (thelp) at (0,3) {$\rthelp$}
edge [pre, bend left] (run)
edge [post, bend right] (run)
edge [post] (temp)
;

\node[transition] (ta) at (4,0) {$\rta$}
    edge [pre] (run)
    edge [post, bend right] (stop)
    ;

\node[transition] (tb) at (4,2) {$\rtb$}
    edge [pre] (temp)
    edge [pre] (run)
    edge [post, bend left] (stop)
    ;

\end{tikzpicture}

%% file: content/upward_bpp.tex

\subsection{BPP Nets}
\label{Subsec:UpBPP}

We establish an exponential upper bound on the size of the finite automata representing the upward closure of BPP net languages. 
Then, we present a family of BPP net languages  for which the minimal finite automata representing their upward closure are of size at least exponential in the size of the input.

\paragraph*{Upper Bound}

Assume that the net $N$ in the Petri net instance $(N,M_0,M_f)$ of size $n$ is a BPP net.

\begin{theorem}
\label{Theorem:UCCompBPP}
    One can construct an FSA of size $O(2^{\poly(n)})$ for $\uc{\lang{N,M_0,M_f}}$.
\end{theorem}
We will show that every minimal word results from a computation whose length is polynomially dependent on the number of transitions and on the number of tokens in the final marking (which may be exponential in the size of the input).  
Let $k$ be a bound on the length of the minimal computations. 
With the same argument as before and using Lemma~\ref{kboundedlanguages},  
we can construct a finite state automaton of size $O({2^{\poly(n)}})$ that accepts $\uc{\langk{k}{N,M_0,M_f}}$.

\begin{proposition}
\label{Proposition:BPPShort}
    Consider a BPP net $N$.
    For every computation $M_0\move{\sigma}M\geq M_f$ there is $M_0\move{\sigma'}M'\geq M_f$ with $\lambda(\sigma') \subword \lambda(\sigma)$ and $\card{\sigma'} \leq \tokencount{M_f}^2 \cdot \card{T}$.
\end{proposition}
The key to proving the proposition is to consider a structure that makes the concurrency among transitions in the BPP computation of interest explicit. 
Phrased differently, we give a true concurrency semantics (also called partial order semantics and similar to Mazurkiewicz traces) to BPP computations. 
Since BPPs do not synchronize, the computation yields a forest where different branches represent causally independent transitions. 
To obtain a subcomputation that covers the final marking, we select from the forest a set of leaves that corresponds exactly to the final marking.
We then show that the number of transitions in the minimal forest that generates the selected set of leaves is polynomial in the number of tokens in the final marking and in the number of transitions.

To make the proof sketch for Proposition~\ref{Proposition:BPPShort} precise, we use (and adapt to our purposes) unfoldings, a true concurrency semantics for Petri nets~\cite{EH08}.
The unfolding of a Petri net is the true concurrency analogue of the computation tree -- a structure that represents all computations.
Rather than having a node for each marking, there is a node for each token in the marking.
To make the idea of unfoldings formal, we need the notion of an \emph{occurrence net}, an unlabeled BPP net $O = (P',T',F')$ that is acyclic and where each place has at most one incoming transition and each transition creates at most one token per place: $\sum_{t'\in T'}F(t', p')\leq 1$ for every $p' \in P'$.
Two elements $x,y \in P'\cup T'$ are \emph{causally related}, $x \causal y$, if there is a path from $x$ to $y$.
We use $\predec{x}=\Set{y\in P'\cup T'}{y \causal x}$ to denote the predecessors of $x\in P'\cup T'$.  
The $\causal$-minimal places are denoted by $\mathit{Min}(O)$.
The initial marking of $O$ is fixed to having one token in each place of $\mathit{Min}(O)$ and no tokens elsewhere. 
So occurrence nets are $1$-safe and we can identify markings with sets of places  $P_1', P_2'\subseteq P'$ and write $P_1'\move{t'}P_2'$. 
To formalize that $O$ captures the behavior of a BPP net $N = (\Sigma,P,T,F,\lambda)$ from marking $M_0$, we define a \emph{folding homomorphism} $h:P' \cup T' \rightarrow P \cup T$ satisfying

\begin{enumerate}[(1)]
    \item
        Initiation: 
        $h(\mathit{Min}(O))=M_0$.
    \item
        Consecution:
        For all $t'\in T'$,   
        $h(\Pred{t'}) = \Pred{h(t')}$,
        and all $p \in P$,
        $(h(\Succ{t'}))(p) = F(h(t'), p)$.
        \\
        Here, $h(P_1') \colon P \to \N $ with $P_1'\subseteq P'$ is a function with $(h(P_1'))(p)=\card{\Set{p'\in P_1'}{h(p')=p}}$.
    \item
        No redundancy: For all $t_1',t_2' \in T'$, with $\Pred{t_1'} = \Pred{t_2'}$ and $h(t_1') = h(t_2')$, we have $t_1' = t_2'$.
\end{enumerate}

\noindent
The pair $(O,h)$ is called a \emph{branching process} of $(N, M_0)$. 
Branching processes are partially ordered by the prefix relation which, intuitively, states how far they unwind the BPP. 
The limit of the unwinding process is the \emph{unfolding} $\text{Unf}(N, M_0)$, the unique (up to isomorphism) maximal branching process. 
It is not difficult to see that there is a one to one correspondence between the firing sequences in the BPP net and the firing sequences in the unfolding.
Note that $\text{Unf}(N, M_0)$ will usually have infinitely many places and transitions, but every computation will only use places up to a bounded distance from $\mathit{Min}(O)$.
With this, we are prepared to prove the proposition.

\begin{proof}[Proof of Proposition~\ref{Proposition:BPPShort}]
    Consider a computation \mbox{$M_0\move{\sigma}M$} with $M \geq M_f$ in the given BPP net $N=(\Sigma, P, T, F, \lambda)$.
    Let $(O, h)$ with $O=(P', T', F')$ be the unfolding $\text{Unf}(N, M_0)$.
    Due to the correspondence in the firing behavior, there is a sequence of transitions $\tau$ in $O$ with $h(\tau) = \sigma$ and $\mathit{Min}(O)\move{\tau}P_1'$ with $h(P_1') = M$.
    Since $M\geq M_f$, we know that for each place $p\in P$, the set $P_1'$ contains at least $M_f(p)$ many places $p'$ with $h(p')=p$.
    We arbitrarily select a set $X_{p}$ of size $M_f(p)$ of such places $p'$ from $P_1'$. 
    Let $X = \bigcup_{p \in P} X_p$ be the union for all $p \in P$.
    
    The computation $\tau$ induces a forest in $O$ that consists of all places that contain a token after firing $\tau$ and their predecessors.
    We now construct a subcomputation by restricting $\tau$ to the transitions leading to the places in $X$.
    Note that the transitions leading to $X$ are contained in $\predec{X}$, which means we can define the subcomputation as $\tau_1=\Prj{\tau}{\predec{X}}$, \ie the projection of $\tau$ onto $\predec{X}$.
    In $\tau_1$, we mark all $\causal$-maximum transitions $t'$ that lead to two different places in $X$.
    Formally, if there are $x, y\in X$ with $t' \in \predec{x} \cap \predec{y}$ and there is no $t'' \in \predec{x} \cap \predec{y}$ with $t' \causal t''$, then we mark $t'$.
    We call the marked $t'$ the \emph{join transitions}.
    
    Assume that $t' \neq t''$ are two  join transitions that occur on the same branch of the forest.
    Note that for two places in $X$, there is either no join transition or a unique one leading to these two places.
    Consequently, $t'$ and $t''$ have to lead to different places of $X$.
    Let $t' t^1 \ldots t^m t''$ be the transitions on the branch in between $t'$ and $t''$.
    We assume that $t'$ and $t''$ are adjacent join transitions, \ie none of the $t^i$ is a join transition.
    
    Since $t', t''$ occur in $\tau_1$, all $t^i$ also have to occur in $\tau_1$. 
    If there are indices $j < k$ such that $t^j = t^k$, we may delete $t^{j+1} \ldots t^{k}$ from $\tau_1$ while keeping a transition sequence that covers $X$.
    It will cover $X$ as none of the deleted transitions was a join transition, \ie we will only lose leaves of the forest that are not in $X$.
    We repeat this deletion process until there are no more repeating transitions between adjacent join transitions.
    Let the resulting transition sequence be $\tau_2$.
    First, note that for any $x \in X$, there are at most $\tokencount{M_f}$ many join transitions on the branch from the corresponding minimal element to $x$:
    In the worst case, for each place in $X \setminus \set{x}$, there is a join transition on the branch, and $\card{X} = \tokencount{M_f}$.
    Between any two adjacent join transitions along such a path, there are at most $\card{T}$ transitions (after deletion).
    Hence, the number of transitions in such a path is bounded by $\tokencount{M_f} \cdot \card{T}$.
    Since we have $\tokencount{M_f}$ many places in $X$, the total number of transitions in $\tau_2$ is bounded by $\tokencount{M_f}^2 \cdot \card{T}$.
\end{proof}

\paragraph*{Lower Bound}

We present a family of BPP net languages for which the minimal FSA representing the upward closure are exponential in the size of the input.
The idea is to rely on the binary encoding of numbers, which allows us to handle $2^n$ using a polynomially sized net.

\begin{lemma}
\label{BPPHard}
    For all numbers $n \in \N$, we can construct a BPP net
    \mbox{$N(n) = (\{a\},P,T,F,\lambda)$}
    and markings $M_0, M_f$ of size polynomial in $n$ such that
    $
    \lang{(N(n),M_0,M_f} = \{ a^{2^n} \}\ .
    $
\end{lemma}

\begin{proof}
    The BPP net $N(n)$ consists of three places $p_0, p_1, p_f$ and two transitions $t,t_a$.
    Transition $t$ is $\varepsilon$-labeled, consumes one token from $p_0$, and creates $2^n$ tokens on $p_1$.
    Transition $t_a$ is labeled by $a$ and moves one token from $p_1$ to $p_f$.
    Formally, $\lambda(t) = \varepsilon, F(p_0,t) = 1, F(t,p_1) = 2^n$, $\lambda(t_a) = a, F(p_1,t_a) = 1, F(t_a,p_f) = 1$.
    All other values for $F$ are $0$.
    The initial marking $M_0$ places one token on $p_1$ and no tokens elsewhere, the final marking $M_f$ requires $2^n$ tokens on $p_f$ and no tokens elsewhere.
    Note that $F$ as well as $M_f$ have polynomially-sized encodings.
    There is a unique covering computation in $N(n)$, namely the computation $M_0 \fire{\sigma} M_f$, where $\sigma = t.\underbrace{ t_a \ldots t_a }_{ 2^n \text{ times}}$.
    Thus, the language of $(N(n),M_0,M_f)$ is as required.
\end{proof}

%% file: content/downward_pn.tex

\section{Downward Closures}
\label{Section:DownwardClosure}

We consider the problem of constructing a finite state automaton accepting the downward closure of a Petri net and a BPP net language, respectively. 
The downward closure often has the property of being a precise description of the system behavior, namely as soon as asynchronous communication comes into play:
If the components are not tightly coupled, they may overlook commands of the partner and see precisely the downward closure of the other's computation. 
As a result, having a representation of the downward closure gives the possibility to design exact or under-approximate verification algorithms.

\begin{compproblem}
    \problemtitle{Computing the downward closure}
    \probleminput{A Petri net instance $(N,M_0,M_f$).}
    \problemquestion{An FSA $A$ with $\lang{A} = \dc{\lang{N,M_0,M_f}}$.}	
\end{compproblem}

\subsection{Petri Nets}

The downward closure of Petri net languages has been shown to be effectively computable in~\cite{HMW10}. 
The algorithm is based on the Karp-Miller tree~\cite{km69}, which can be of non-primitive recursive size. 
We now present a family of Petri net languages that are already downward closed and for which the minimal finite automata have to be of non-primitive recursive size in the size of the input. 
Our result relies on a construction due to Mayr and Meyer~\cite{MayrMeyer1981}. 
It gives a family of Petri nets whose computations all terminate but, upon halting, may have produced Ackermann many tokens on a distinguished place.

We first recall the definition of the Ackermann function.

\begin{definition}
    The Ackermann function is defined inductively as follows:

    $\begin{tabu}{l@{}cll}
    \Acker_0&(x) &=& x+1
    \\
    \Acker_{n+1}&(0) &=& \Acker_n(1)
    \\
    \Acker_{n+1}&(x+1) &=& \Acker_n(\Acker_{n+1}(x))
    \ .\\
    \end{tabu}$
\end{definition}
\begin{lemma}
\label{ackermann}
    For all $n, x \in \N$, there is a Petri net $N(n) = (\{a\},P,T,F,\lambda)$
    and markings $M_0^{(x)}, M_f$ of size polynomial in $n + x$ such that
    $
        \lang{N(n),M_0^{(x)},M_f} = \set{ a^k \mid k \leq \Acker_n(x) }.
    $
\end{lemma}
Our lower bound is an immediate consequence of this lemma.

\begin{theorem}
\label{Theorem:DCComp}
    There is a family of Petri net languages for which the minimal finite automata representing the downward closure are of non-primitive recursive size.
\end{theorem}
This hardness result relies on a weak computation mechanism of very large numbers that is unlikely to show up in practical examples. 
The SRE inclusion problem studied in the following section can be understood as a refined analysis of the computation problem for downward closures.

It remains to prove Lemma~\ref{ackermann}.
We start by defining a preliminary version of the required nets.
The construction is inductive and imitates the definition of the Ackermann function.

\begin{definition}
    We define the Petri net $\mathit{AN}_0$ to be
    \begin{align*}
        \mathit{AN}_0 &= (\set{a},P^0,T^0,F^0,\lambda^0) \quad \text{ with }\\
    P^0 &= \set{ \pin^0, \pout^0, \pstart^0, \pstop^0, \pcopy^0 },\\
    T^0 &= \set{ \tstart^0, \tstop^0, \tcopy^0 }, \\
    \lambda^0 (t) &= \varepsilon \text{ for all } t \in T^0 \ .
    \end{align*}
    The flow relation is given by Figure~\ref{pic1}, where each edge carries a weight of $1$.
    
    \begin{figure}[b]
        \centering
        \scalebox{0.9}{\input{tikz/ackermann_basecase.tikz}}
        \caption{\!\!\textbf{\sffamily .} The flow relation of the Petri net $\mathit{AN}_0$.}
        \label{pic1}
    \end{figure}
    \noindent
    For $n \in \N$, we define $\mathit{AN}_{n+1}$ inductively by
    \begin{align*}
    \mathit{AN}_{n+1}
    &= (\set{a},P^{n+1},T^{n+1},F^{n+1},\lambda^{n+1})
    \quad \text{ with }
    \\
    P^{n+1}
    &= P^n \cup \{ \pin^{n+1}, \pstart^{n+1}, \pcopy^{n+1}, \pout^{n+1}, \pstop^{n+1},\pswap^{n+1}, \ptmp^{n+1}, \} 
    \\
    T^{n+1}
    &= T^n \cup\{ \tstart^{n+1}, \tcopy^{n+1}, \tstop^{n+1},\trestart^{n+1},\tin^{n+1}, \tswap^{n+1}, \ttmp^{n+1} \},
    \\
    \lambda^{n+1}(t) &= \varepsilon \text{ for all }t \in T^{n+1} .
    \end{align*}
    The flow relation is given by Figure~\ref{pic2}, where again each edge carries a weight of $1$.
    
    \begin{figure*}[ht!]
        \centering
        \scalebox{0.9}{
        \input{tikz/ackermann_step.tikz}
        }
        \caption{\!\!\textbf{\sffamily .} The flow relation of the Petri net $\mathit{AN}_{n+1}$.}
        \label{pic2}
    \end{figure*}

    Let us furthermore define for each $x \in \N$ the marking $M_0^{(x)}$ of $\mathit{AN}_n$ that places one token on $\pstart^{n}$, $x$ tokens on $\pin^{n}$ and no token elsewhere.
\end{definition}
We prove that $\mathit{AN}_{n}$ indeed can weakly compute $\Acker_n(m)$.

\begin{lemma}
    \label{ackerman-computation}
    For all $n,x \in \N$:\\
    (1) There is $M_0^{(x)} \move{\sigma} M$ of $\mathit{AN}_n$ such that $M(\pout^n) = \Acker_n(x)$, $M(\pstop^n) = 1$.\\
    (2) There is no computation starting in $M_0$ that creates more than $\Acker_n(x)$ tokens on $\pout^n$.
\end{lemma}

\begin{proof}
    We proof both statements simultaneously by induction on $n$.
    
    \noindent
    \textbf{\sffamily Base case}, $n = 0$:
    We have $\Acker_0(x) = x+1$.
    The only transition that is enabled in $\mathit{AN}_0$ is the starting transition $\tstart^0$.
    Firing it leads to one token on the copy place $\pcopy^0$.
    Now we can fire the copy transition $\tcopy^0$ $x$ times, leading to $x$ tokens on $\pout^0$.
    Finally, we fire the stopping transition $\tstop^0$, which leads to one token on $\pstop^0$ and in total $x+1$ tokens on $\pout^0$.
    This is the computation maximizing the number of tokens on $\pout$.
    Firing $\tcopy^0$ less than $x$ times or not firing $\tstop^0$ leads to less tokens on $\pout^0$.
    
    \noindent
    \textbf{\sffamily Inductive step, $n \mapsto n+1$}:
    Initially, we can only fire the starting transition $\tstart^{n+1}$, creating one token on $\pin^n$ and one token on $\pstart^n$.
    We can now execute the computation of $\mathit{AN}_n$ that creates $\Acker_n(1)$ tokens on $\pout^n$ and one token on $\pstop^n$, which exists by induction.
    We consume one token from $\pin^{n+1}$ and the token on $\pstop^n$ to create a token on $\pswap^{n+1}$ using $\tin^{n+1}$.
    This token allows us to swap all $\Acker_n(1)$ tokens from $\pout^{n}$ to $\pin^{n}$ using $\tswap^{n+1}$.
    After doing this, we move the token from $\pswap^{n+1}$ to $\tstart^n$ using the restart transition $\trestart^{n+1}$.
    We iterate the process to create
    \[
    \Acker_n (\Acker_n (1)) = \Acker_n (\Acker_{n+1} (0)) = \Acker_{n+1} (1)
    \]
    tokens on $\pout^{n+1}$, which we can then swap again to $\pin^{n}$.
    
    We iterate this process $x$-times, which creates $\Acker_{n+1} (x)$ tokens on $\pout^n$, since
    \[
    \Acker_{n+1} (x)
    =
    \underbrace{ \Acker_n ( \ldots \Acker_n }_{x+1 \text{ times}} (1) )
    \ .
    \]
    Note that it is not possible to create more than $\Acker_{n+1} (x)$ on $\pout ^n$.
    Having $y$ tokens on $\pin^n$, we cannot create more than $\Acker_{n} (y)$ tokens on $\pout^n$ by induction.
    Furthermore, if we do not execute $\tswap$ as often as possible, say we leave $y'$ out of $y$ tokens on $\pout^n$, we end up with
    \mbox{$\Acker_n (y - y') + y' \leq \Acker_n (y)$}
    tokens on $\pout^n$, since
    \[
    \Acker_n (k) + k' \leq \Acker_n(k + k')
    \]
    for all $k, k' \in \N$.
    Finally, we move the token on $\tstop^{n+1}$ to $\ptmp^{n+1}$ using $\ttmp^{n+1}$, and then move all $\Acker_n (x)$ tokens to $\pout^{n+1}$ using $\tcopy^{n+1}$.
    We end the computation by firing $\tstop^{n+1}$ to move the token on $\ptmp^{n+1}$ to $\pstop^{n+1}$.
    We now have one token on $\pstop^{n+1}$ and $\Acker_n (x)$ tokens on $\pout^{n+1}$ as required.
    
    This maximizes the number of tokens on $\pout^{n+1}$:
    As already argued, we cannot create more than $\Acker_n (x)$ tokens on $\pout ^n$, and firing $\tcopy^{n+1}$ less than $\Acker_n (x)$ times will lead to less tokens on $\pout^{n+1}$.
\end{proof}
We are now prepared to tackle the proof of Lemma~\ref{ackermann}.

\begin{proof}[Proof of Lemma~\ref{ackermann}]
    Let $n \in \N$.
    We define $N(n)$ to be the Petri net that is obtained from $\mathit{AN}_n$ by adding a place $\text{final}$ and a transition $t_\mathit{final}$ that is labeled by $a$ and moves one token from $\pout^n$ to $\text{final}$.is defined as before.
    The final marking $M_f$ is zero on all places.
    
    By Lemma~\ref{ackerman-computation}, we can create at most $\Acker_n(x)$ tokens on $\pout^n$.
    We can then move a part of these tokens to $\text{final}$, producing up to $\Acker_n(x)$ many $a$s in the process.
    This proves $\lang{N(n),M_0^{(x)},M_f} = \set{ a^k \mid k \leq \Acker_n(x) }$.
    
    Note that the size of $N(n)$ is a constant plus the size of $\mathit{AN}_n$, which is linear in $n$.
    The final and initial marking are linear in $n + \log x$.
\end{proof}

%% file: tikz/ackermann_basecase.tikz
\begin{tikzpicture}[>=stealth',shorten >=1pt,auto,node distance=1.5cm]

\node[place] (A1) [minimum size=1.1cm] {$\pin^0$};

\node[transition] (B1) [right of=A1, node distance=3.5cm, inner sep = 2pt] {$\tcopy^0$}
    edge [pre] (A1);

\node[place] (A4) [ right of=B1, minimum size=1.1cm, node distance=3.5cm] {$\pout^0$}
    edge [pre] (B1);

\node[place] (A2) [ below of=A1, minimum size=1.1cm,  node distance = 1.5cm] {$\pstart^0$};

\node[transition] (B2) [right of=A2, xshift=0.25cm, inner sep = 2pt] {$\tstart^0$}
    edge [pre] (A2);

\node[place] (A3) [below of=B1, minimum size=1.1cm, node distance = 1.5cm] {$\pcopy^0$}
    edge [pre] (B2)
    edge [pre, bend left] (B1)
    edge [post,bend right] (B1);

\node[transition] (B3) [right of=A3, xshift=0.25cm, inner sep = 2pt] {$\tstop^0$}
    edge [pre] (A3)
    edge [post] (A4) ;

\node[place] (A5) [ below of=A4, minimum size=1.1cm, node distance = 1.5cm] {$\pstop^0$}
    edge [pre] (B3);

\end{tikzpicture}

%% file: tikz/ackermann_step.tikz
\begin{tikzpicture}[>=stealth',shorten >=1pt,auto,node distance=1.3cm]

\node[place] (A1) [minimum size=1.5cm] {$\pin^{n+1}$};

\node[place] (A2)  [below of=A1, node distance = 2cm, minimum size=1.5cm] {$\pstart^{n+1}$};

\node[draw, thick, minimum width=3cm, minimum height=3cm, inner sep=0pt, outer sep=0pt, right of=A1, node distance=4.7cm, yshift=-1cm, ]{$\mathit{AN}_n$};
    
\node [place, fill=white, right of=A1, node distance=3.5cm, minimum size=1.2cm] (AA1){$\pin^{n}$};

\node [place, fill=white, node distance=3.5cm, minimum size=1.2cm] (AA2) [right of=A2] {$\pstart^{n}$};

\node [place, fill=white] (AB1) [right of=AA1, node distance = 2.5cm, minimum size=1.2cm] {$\pout^{n}$};

\node [place, fill=white] (AB2) [right of=AA2, node distance = 2.5cm, minimum size=1.2cm] {$\pstop^{n}$};

\node[transition] (B1) [right of=A2, xshift=0.5cm, inner sep = 2pt] {$\tstart^{n+1}$}
    edge [pre] (A2)
    edge [post] (AA1)
    edge [post] (AA2);

\node[transition] (B2) [above of=AA1, xshift=1cm, inner sep = 2pt] {$\tswap^{n+1}$}
    edge [pre, bend left=40] (AB1)
    edge [post, bend right=40] (AA1);

\node[place] (A3) [above of=B2, node distance = 1.5cm, minimum size=1.5cm] {$\pswap^{n+1}$}
    edge [pre, bend left] (B2)
    edge [post, bend right] (B2);

\node[transition] (B3) [left of=B2, node distance = 2.5cm, inner sep = 2pt] {$\trestart^{n+1}$}
    edge [pre,bend left] (A3)
    edge [post, bend right] (AA2);

\node[transition] (B4) [above of=A3, node distance = 1.5cm, inner sep = 2pt] {$\tin^{n+1}$}
    edge [pre, bend left=60] (AB2)
    edge [pre, out=180,in=90] (A1)
    edge [post] (A3);

\node[transition] (B4) [right of=AB1, node distance=3.5cm, inner sep = 2pt] {$\tcopy^{n+1}$}
    edge [pre] (AB1);

\node[transition] (B5) [right of=AB2, xshift=0.4cm, inner sep = 2pt] {$\ttmp^{n+1}$}
    edge [pre] (AB2);

\node[place] (A4) [right of=AB2, node distance=3.5cm, minimum size=1.5cm] {$\ptmp^{n+1}$ }
    edge [pre, bend left] (B4)
    edge [pre] (B5)
    edge [post, bend right] (B4);

\node[transition] (B6) [right of=A4, xshift=0.4cm, inner sep = 2pt] {$\tstop^{n+1}$}
    edge [pre] (A4);

\node[place] (A5) [right of=A4, node distance=3.5cm, minimum size=1.5cm] {$\pstop^{n+1}$}
    edge[pre] (B6);

\node[place] (A6) [right of=B4, node distance=3.5cm, minimum size=1.5cm] {$\pout^{n+1}$}
    edge [pre] (B4);

\end{tikzpicture}

%% file: content/downward_bpp.tex

\subsection{BPP Nets}
\label{sec:BPPDC}

We prove an exponential upper bound on the size of the finite automata representing the downward closure of BPP languages. 
Then, we present a family of BPP languages for which the minimal finite automata representing their downward closure are exponential in the size of the input BPP nets.

\paragraph*{Upper Bound}

Assume that the net $N$ in the Petri net instance $(N,M_0,M_f)$ of size $n$ is a BPP net.

\begin{theorem}
\label{DCCompBPP}
    We can construct a finite automaton of size $O(2^{\poly(n)})$ for $\dc{\lang{N,M_0,M_f}}$.
\end{theorem}
The key insight
for simulating $N$ by a finite automaton
is the following:
If
during a firing sequence
a marking occurs that has more than $c$ tokens (where $c$ is specified below) in some place $p$, then there has to be a \emph{pump}, a subsequence of the firing sequence that can be repeated to produce arbitrarily many tokens in $p$. 
The precise statement is this, where we use $m = \mathit{max}(F)$ to refer to the maximal multiplicity of an edge.
\begin{lemma}
\label{bpppump}
    Let $M_0\move{\sigma}M$ such that for some place $p\in P$, we have $M(p)>c$ with 
    \[
        c = \tokencount{M_0} \cdot (\card{P}\cdot m)^{(|T| + 1)}
        \ .
    \]
    Then for each $j \in \N$, there is $M_0\move{\sigma_j}M_j$ such that\\
    (1)~$\sigma\subword\sigma_j$,\
    (2)~$M\leq M_j$, and \ 
    (3)~$M_j(p)>j$.
\end{lemma}

\begin{proof}
    We consider the unfolding $(O, h)$ of $N$, as defined in Subsection~\ref{Subsec:UpBPP}.
    Let $\sigma'$ be a firing sequence of~$O$ induced by $\sigma$, i.e. $h(\sigma') = \sigma$ (where $h$ is extended to sequences of transitions in the obvious way), and $\sigma'$ can be fired from the marking $M_0'$ of $O$ corresponding to $M_0$.
    
    Executing $\sigma'$ leads to a marking $M'$, such that the sum of tokens in places $p'$ of $O$ with $h(p') = p_N$ is equal to $M(p_N)$ for each place $p_N$ of $N$.
    In particular, this holds for the place $p$ on which we exceed the bound $c$:
    \[
    \sum_{\substack{p' \in P',\\  h(p') = p}} M' (p') = M(p) > c
    \ .
    \]
    Recall that $O$ is a forest, and each tree in it has a minimal place $r' \in \mathit{Min}(O)$ that corresponds to a token assigned to a place of $N$ by $M_0$, \ie $M_0' (r') = 1$.
    We fix the root node $r$ of the tree $\calT$ with a maximal number of leaves that correspond to place $p$ (called \emph{$p$-leaves} in the following), \ie the root node $r$ such that the number of places $p'$ with $h(p') = p$, $M' (p')= 1$ in the corresponding tree is maximal.
    Note that the number of $p$-leaves in this tree is at least
    \[
        c_1 = \frac{c}{\tokencount{M_0}} = (\ell \cdot m)^{(|T| + 1)}
        \ .
    \]
    Since there are only $\tokencount{M_0}$ many trees in the forest, if all trees have strictly less than $c_1$ many $p$-leaves, the whole forest cannot have $c$ many $p$-leaves.
    
    We now consider the subtree $\calT_p$ of $\calT$ that is defined by the $p$-leaves in $\calT$, \ie the tree one gets by taking the set of $p$-leaves $X$ in $\calT$ and all places and transitions $\predec{X}$ that are their predecessors.
    This tree has the following properties:
    \begin{enumerate}[(i)]
        \item
        Its leaves are exactly the $p$-leaves in $\calT$.
        \item
        Each place in it has out-degree $1$ if it is not a $p$-leaf.
        That the out-degree is at most $1$ is clear by how $O$ was defined: Since each place only carries at most one token, it can be consumed by at most one transition during the run, and we don't consider the transitions that are not fired in the run in $\calT_p$.
        That the out-degree of the places that are not $p$-leaves is exactly $1$ is because we only consider transitions leading to $p$-leaves in $\calT_p$.
        \item
        Each transition has out-degree at most $m \cdot \ell$.
        (In the worst case, each transition creates $m$ tokens in each of the $\ell$ places, which is modeled in $O$ by having one place for each token that is produced.)
    \end{enumerate}
    \vspace*{2pt}
    
    \noindent
    We will call a transition in $\calT_p$ of out-degree at least $2$ a \emph{join-transition}, since it joins (at least) two branches of the tree that lead to a $p$-leaf.
    Our goal is to show that there is a branch of $\calT_p$ in which at least $|T| + 1$ many join-transitions occur.
    
    \vspace*{4pt}
    
    \noindent\textbf{\sffamily Claim:} Let $\calT$ be a tree with $x$ leaves in which all nodes have out-degree at most $k$.
    Then $\calT$ has a branch with at least $\log_k x$ nodes of out-degree at least $2$.
    
    Towards a proof of the claim, assume that the maximal number of nodes of out-degree greater than $2$ in any branch of the tree is $h < \log_k x$.
    To maximize the number of leaves, we assume that the number of such nodes is exactly $h$ in every branch, and all nodes (but the leaves) have out-degree $k$.
    The number of leaves in this tree is $k^h$, but since $h < \log_k x$, this is less than $x$: $k^h < k^{\log_k x} = x$.
    
    Instantiating the claim for $x = c_1$ and $k = m \cdot |M_0|$ yields that $\calT_p$ has a branch with at least $c_2 = \log_{m \cdot \tokencount{M_0}} c_1$ many join-transitions, and by the definition of $c$, \mbox{$c_2 = |T| + 1$.}
    Since the original BPP net has only $|T|$ many different transitions, there have to be join-transitions $\tau$ and $\tau'$ in the same branch with $h(\tau) = h(\tau') = t$ for some transition $t$ of the original net.
    
    Since $\tau$ was a join-transition, it has at least two child branches $b_1$ and $b_2$ that lead to a $p$-leaf.
    We assume without loss of generality that $b_1$ is the branch on which $\tau'$ occurs, and $b_2$ is another branch.
    We consider the sequence of transitions $\sigma_p^{\mathit{pump}}$ that occur on $b_1$ when going from $\tau$ to $\tau'$ (including $\tau$, not including $\tau'$).
    We also consider the sequence of transitions $\sigma_p^{\mathit{gen}}$ that occur on $b_2$ when going from $\tau$ to a $p$-leave (not including $\tau$).
    Let $\sigma^{\mathit{pump}}$ and $\sigma^{\mathit{gen}}$ be the corresponding sequences in the original BPP net, \ie $\sigma^{\mathit{pump}} = h(\sigma_p^{\mathit{pump}}), \sigma^{gen} = h(\sigma_p^{gen})$.
    
    We now modify the run $\sigma$ in the original net to obtain the desired amount of $j$ tokens.
    We decompose $\sigma = \sigma_1 . t .\sigma_2$, where $t$ is the transition that corresponds to $\tau$ in $\calT_p$.
    We extend the run to
    \[
    \sigma_j
    =
    \sigma_1 .
    \underbrace{\sigma^{\mathit{pump}} \ldots \sigma^{\mathit{pump}}}_{j \text{ times}}
    . t. \sigma_2 .
    \underbrace{\sigma^{\mathit{gen}} \ldots \sigma^{\mathit{gen}}}_{j \text{ times}}
    \ .
    \]
    Obviously, $\sigma$ is a subsequence of $\sigma_j$, and therefore satisfies the required Property~(1).
    We have to argue why $\sigma_j$ is a valid firing sequence,
    why firing $\sigma_j$ leads to a marking greater than $M$ (Property~(2)) and why it generates at least $j$ tokens in place $p$ (Property~(3)).
    
    The latter is easy:
    $\sigma^{\mathit{gen}}$ corresponds to a branch of $\calT_p$ leading to a $p$-leaf, \ie firing it creates one additional token in $p$ that will not be consumed by another transition in $\sigma$.
    
    Note that up to $\sigma_1$, $\sigma$ and $\sigma_j$ coincide.
    Since $t$ could be fired after $\sigma_1$, $\sigma^{\mathit{pump}}$ can be fired after $\sigma_1$: $t$ corresponds to transition $\tau$ in $O$, and so does the first transition in $\sigma^{\mathit{pump}}$.
    Since $\sigma^{\mathit{pump}}$ was created from a branch in the tree $\calT_p$, each transition will consume the token produced by its immediate predecessor.
    The last transition creates a token in the place that feeds transition $\tau'$, but $h(\tau) = t = h(\tau')$, so after firing $\sigma^{\mathit{pump}}$, we can fire $t$ again. (Either as the first transition of the next $\sigma^{\mathit{pump}}$, or after all pumps have been fired.)
    Furthermore, $t$ corresponds to the join-transition $\tau$, i.e. it does create a token in the place that is the starting point of $\sigma^{\mathit{gen}}$.
    This token will not be consumed by any other transition in the sequence $t. \sigma_2$, so after firing $\sigma^{\mathit{pump}}$ $j$ times, we can indeed fire $\sigma^{\mathit{gen}}$ $j$ times.
    
    As argued above, firing $\sigma^{\mathit{gen}}$ and $\sigma^{\mathit{pump}}$ has a non-negative effect on the marking, so the marking one gets by firing $\sigma'$ is indeed greater than the marking $M$.
\end{proof}
It remains to use Lemma~\ref{bpppump} to prove Theorem~\ref{DCCompBPP}.

\begin{proof}[Proof of Theorem~\ref{DCCompBPP}]
    We will state the construction of the automaton, prove its soundness and that the size of the automaton is as required.
    The automaton for $\dc{\lang{N,M_0,M_f}}$ is the state space of $N$ with token values beyond $c$ set to $\omega$.
    For every transition, we also have an $\varepsilon$-variant to obtain the downward closure. 

    More formally, $A = (\Sigma, Q,\to_A,\qinit,F)$ is defined as follows:
    Its set of states is $Q = P \to (\zeroto{c} \cup \set{\omega})$, where
    \mbox{$c = \tokencount{M_0} \cdot (\card{P}\cdot m)^{(|T| + 1)}$} as in Lemma~\ref{bpppump}.
    This means each state is a marking that will assign to each place a number of tokens up to $c$ or $\omega$.
    For each transition $t$ of the BPP net $N$ and each state $q \in Q$ such that $q(p) \geq F(p,t)$ (where we define $\omega > k$ to be true for all $k \in \N$), $\to_A$ contains two transitions $(q,\lambda(t),q')$ and $(q,\varepsilon,q')$.
    Here, $q'$ is defined by
    \[
        q'(p) = (q(p) \ominus F(p,t)) \oplus F(t,p)
        \ ,
    \]
    for all $p \in P$, where $\oplus$ and $\ominus$ are variants of $+$ and $-$ that treat $\omega$ as infinity:
    $x \oplus y = x+y$ if $x+y<c$, $x\oplus y=\omega$ otherwise. Similarly, $x\ominus y = x-y$ if $x\neq \omega$.
    Note that if $t$ was already labeled by $\varepsilon$, the two transitions coincide.
    The initial state is defined by $\qinit(p) = M_0(p)$ for all $p \in P$.
    A state $q \in F$ is final if it covers the final marking $M_f$ of $N$, \ie $q(p) \geq M_f(p)$ for all places $p$.
    Again, we assume $\omega > k$ to hold for all $k \in \N$.
    
    We prove that indeed $\lang{A} = \dc{\lang{N,M_0,M_f}}$ holds.
    First assume $w \in \dc{\lang{N,M_0,M_f}}$.
    Then there is a computation $\pi = M_0 \move{\sigma} M$ of $N$ such that $M \geq M_f$ and $w \subword \lambda(\sigma)$.
    We can delete transitions in $\sigma$ to obtain a sequence of transitions $\tau$ with $\lambda(\tau) = w$.
    (One may not be able to fire $\tau$.)
    We construct a run $\rho$ of the automaton $A$ starting in $q_{0}$ by replacing transitions in $\sigma$ by a corresponding transition of the automaton.
    For the transitions $t$ present in $\tau$, we pick the variant of the transitions labeled by $\lambda(t)$, for the ones not present in $\tau$, we pick the $\varepsilon$-labeled variant.
    Note that $\rho$ is a valid run of $A$ because $\pi$ was a computation of $N$.
    The run $\rho$ ends in a state $q_\rho$ such that for each place $p$, either $q_\rho(p) = M(p)$ holds, or $q(p) = \omega$.
    Since $M \geq M_f$, this means that $q_\rho$ is final.
    We have constructed an accepting run of $A$ that produces the word $w$.
    
    Now assume that $w \in \lang{A}$ is a word accepted by the automaton.
    Let $\rho$ be an accepting run and let $q_0, q_1, \ldots, q_s$ be the states occurring during $\rho$.
    We prove that there is a computation $\pi = M_0 \move{\sigma} M$ of $N$ such that $M \geq M_f$ and $w \subword \lambda(\sigma)$.
    Assume the final state $q_s$ does not assign $\omega$ to any place, $q_s (p) \neq \omega$ for all $p \in P$.
    Note that in this case, we have $q_i (p) \neq \omega$ for all $i \in \zeroto{s}$ and all $p \in P$, since $q_i(p) = \omega$ implies $q_{j} (p) = \omega$ for all $j \in \fromto{i}{s}$.
    In this case, we can easily construct a sequence of transitions $\sigma$ corresponding to $\rho$:
    For each transition in $\rho$, we take the corresponding transition of $N$.
    Note that the transition in $\rho$ can be labeled by $\varepsilon$ while the transition of $N$ is not labeled by $\varepsilon$.
    Still, $w \subword \lambda(\sigma)$ will hold.
    Furthermore, $q_s$ is a marking for $N$ (since $\omega$ does not occur), and since $q_s$ was final, $q_s \geq M_f$ has to hold.
    
    Now assume that there is a unique place $p$ such that \mbox{$q_s (p) = \omega$}.
    Let $i \in \zeroto{s}$ be the first index such that $q_i(p) = \omega$.
    We decompose the run $\rho = \rho_1.\rho_2$, where $\rho_1$ is the prefix that takes the automaton from state $q_0$ to $q_{i}$.
    As in the previous case, we may obtain sequences of transitions $\sigma_1$ and $\sigma_2$ that correspond to $\rho_1$ and $\rho_2$.
    In particular, we have $w \subword \lambda(\sigma_1).\lambda(\sigma_2)$.
    The first sequence $\sigma_1$ is guaranteed to be executable, \ie $\pi_1 = M_0 \move{\sigma_1} M_1$ is a valid computation for some $M_1$.
    Since  $q_i (p) = \omega$, the transition relation of the automaton guarantees that $M_1 (p) > c$.
    
    It might not be possible to fire $\sigma_2$ from $M_1$, because $\sigma_2$ may consume more than $c$ tokens from place $p$.
    Let
    \[
        d = \sum_{j \in \oneto{\card{\sigma_2}}} F(p,t_j) + M_f (p)
        \ .
    \]
    where $\sigma_2=t_1. t_2 \ldots t_{|\sigma_2|}$.  The number $d$ is certainly an upper bound for the number of tokens needed in place $p$ to be able to fire $\sigma_2$ and end up in a marking $M_2$ such that $M_2(p) \geq M_f (p)$.
    We apply Lemma~\ref{bpppump} to obtain a supersequence $\sigma_1'$ of $\sigma_1$ with $M_0 \move{\sigma_1'} M_1'$ where $M_1' \geq M_1$ and $M_1 (p) \geq d$.
    Now consider the concatenation $\sigma = \sigma_1' . \sigma_2$.
    Since the marking $M_1'$ has enough tokens in place $p$, $\sigma$ is executable and $M_0 \move{\sigma} M$, where $M \geq M_f$.
    
    If the final state $q_p$ assigns $\omega$ to several places, the above argumentation has to be applied iteratively to all such places.
    
    It remains to argue argue that the size of the automaton is in $\bigO{2^{\poly (n)}}$.
    The size of the automaton is certainly polynomial in its number of states $\card{Q}$.
    We have
    \begin{align*}
    \card{Q}
    &= \card{ P \to (\zeroto{c} \cup \set{\omega}) }
    = \card{\zeroto{c} \cup \set{\omega}}^\ell
    = (c + 2)^\ell\\
    &= \left(\tokencount{M_0}(\card{P}\cdot m)^{(|T| + 1)} + 2 \right)^\ell
    \leq \left( (\ell \cdot 2^n) (\ell\cdot 2^n)^{(|T| + 1)} + 2 \right)^\ell\\
    &\leq \left( (2^{(n+1)})^{(n + 2)} + 2 \right)^n
    = \left( 2^{(n+1) \cdot (n+2)} + 2 \right)^n\\
    &\leq (2^{(n+1) \cdot (n+2)})^n \cdot 2^n 
    \leq 2^{(n+1) \cdot (n+2) \cdot n + 1}
    \in \bigO{2^{\poly(n)}}
    \ .
    \end{align*}
\end{proof}

\paragraph*{Lower Bound}

Consider the family of BPP nets from Lemma~\ref{BPPHard} with $\lang{N(n),M_0,M_f} = \{ a^{2^n} \}$ for all $n\in\N$. 
The minimal finite state automata recognizing the downward closure $\{ a^{i} \mid i \leq 2^n \}$ has at least $2^n$ states.

%% file: content/sre_downward_pn.tex

\section{SRE Inclusion in Downward Closure}
\label{Section:SRED}

The downward closure of a Petri net language is hard to compute.  
We therefore propose to under-approximate it by an SRE as follows. 
Assume we have a heuristic coming up with a candidate SRE that is supposed to be an under-approximation in the sense that its language is included in the downward closure of interest. 
The problem we study is the algorithmic task of checking whether the inclusion indeed holds.
If so, the SRE provides reliable (must) information about the system's behavior, behavior that is guaranteed to occur. 
This information is useful for finding bugs.

\begin{problem}
	\problemtitle{SRE Inclusion in Downward Closure (\sred)}
	\probleminput{A Petri net instance $(N,M_0,M_f)$, an SRE $sre$.}
	\problemquestion{$\lang{sre} \subseteq \dc{\lang{N,M_0,M_f}} $?}
	
\end{problem}

\subsection{Petri Nets}

\begin{theorem}
\label{Theorem:SREIPN}
    \sred\ is $\EXPSPACE$-complete for Petri nets.
\end{theorem}
Hardness is due to the hardness of coverability~\cite{Lipton}.

\begin{lemma}
\label{Lemma:SREIPNcomplete}
    \sred\ is $\EXPSPACE$-hard for Petri nets.
\end{lemma}

\begin{proof}
    We reduce the $\EXPSPACE$-complete coverability problem for Petri nets.
    Given an Petri net instance $(N,M_0,M_f)$, where $N$ is an unlabeled net, we equip $N$ with the labeling $\lambda(t) = \varepsilon$ for all transitions $t$.
    We have that $M_f$ is coverable from $M_f$ if and only if $\lang{N,M_0,M_f} = \set{\varepsilon}$.
    If $M_f$ is not coverable, the language is empty)
    Thus, we have that $M_f$ is coverable iff $\set{ \varepsilon } \subseteq \dc{\lang{N,M_0,M_f}}$.
    To conclude the proof, note that $\set{ \varepsilon }$ is the language of the SRE $\emptyset^*$.
\end{proof}
For the upper bound, we take inspiration from a recent result of Zetzsche~\cite{Zetzsche15}. 
He has shown that, for a large class of models, computing the downward closure is equivalent to deciding an unboundedness problem. 
We use a variant of this problem that comes with a complexity result. 
The \emph{simultaneous unboundedness problem for Petri nets} (\suppn) is, given a Petri net $N$, an initial marking $M_0$, and a subset $X \subseteq P$ of places, decide whether for each $n \in \N$, there is a computation $\sigma_n$ such that $\TEnable{M_0}{\sigma_n}{M_{\sigma_n}}$ with $M_{\sigma_n}(p) \geq n$ for all places $p \in X$. In \cite{Demri2013}, Demri has shown that this problem is $\EXPSPACE$-complete.

\begin{theorem}[\cite{Demri2013}]
    \suppn\ is $\EXPSPACE$-complete.
\end{theorem}
We turn to the reduction of the inclusion problem \sred\ to the unboundedness problem \suppn.  
Since SREs are choices among products, an inclusion \mbox{$\lang{\sre}\subseteq \dc{\lang{N, M_0, M_f}}$} holds iff $\lang{\product}\subseteq \dc{\lang{N, M_0, M_f}}$ holds for all products $\product$ in $\sre$.  
Since $\dc{\lang{N, M_0, M_f}}$ is downward closed, we can further simplify the products by removing choices. 
Fix a total ordering on the alphabet $\Sigma$.
Such an ordering can be represented by a word $w_{\Sigma}$.
We define the \emph{linearization operation} that takes a product and returns a regular expression:
\begin{align*}
    \linof{a+\varepsilon} &=a
    &\linof{a}&=a\\
    \linof{\Gamma^*}&=(\Prj{w_{\Sigma}}{\Gamma})^*
    &\linof{\product_1.\product_2} &= \linof{\product_1}.\linof{\product_2}\ .
\end{align*}
For example, if $\Sigma=\set{a, b, c}$ and we take $w_{\Sigma}=abc$, then $\product=(a+c)^*(a+\varepsilon)(b+c)^*$ is turned into $\linof{\product}=(ac)^*a(bc)^*$. 
The discussion justifies the following lemma.
\begin{lemma}
\label{Linearization}
    $\lang{\sre}\subseteq \dc{\lang{N, M_0, M_f}}$
    if and only if for all products $\product$ in $\sre$ we have
    \mbox{$\lang{\linof{\product}}\subseteq \dc{\lang{N, M_0, M_f}}$}.
\end{lemma}

\begin{proof}
    $\lang{\linof{\product}} \subseteq \lang{\sre}$ holds, so one direction is clear.
    
    For the other direction, we show that every word in $\lang{\sre}$ is a subword of a word in $\lang{\linof{\product}}$.
    From this, if $\lang{\linof{\product}}$ is included in the downward closure, then all its subwords will be contained in the downward closure.
    In particular, all words in $\lang{\sre}$ will be contained in the downward closure.
    Towards proving that every word in $\lang{\sre}$ is a subword of a word in $\lang{\linof{\product}}$, note that for any word in $\lang{(a+\varepsilon)}$, the letter $a$ may or may not occur, while in $\linof{(a+\varepsilon)} = a$, it is forced to occur.
    Furthermore, given $v \in \Gamma^*$, we have that $v$ is a subword of $\Prj{w_{\Sigma}}{\Gamma}^{\card{v}}$ by dropping in each iteration all but one letter.
    Therefore, all words in $\Gamma^*$ are subwords of $\linof{\Gamma^*}$.
    If we combine those two insights, the desired statement follows.
\end{proof}
With Lemma~\ref{Linearization} at hand, it remains to check $\lang{\linof{\product}}\subseteq \dc{\lang{N, M_0, M_f}}$ for each product.
To this end, we reduce this check to \suppn.
We first understand $\linof{\product}$ as a Petri net $N_{\linof{\product}}$..
We modify this Petri net by adding one place $p_{\Gamma}$ for each block $(\Prj{w_{\Sigma}}{\Gamma})^*=a_i\ldots a_j$.
Each transition that repeats or leaves the block (the ones labeled by $a_j$) is modified to generate a token in $p_{\Gamma}$.
As a result, $p_{\Gamma}$ counts how often the word $\Prj{w_{\Sigma}}{\Gamma}$ has been executed.

The second step is to define an appropriate product of $N_{\linof{\product}}$ with the Petri net of interest. 
Intuitively, the product synchronizes with the downward closure  of $N$.

\begin{definition}
    Consider two Petri nets $N_i = (\Sigma, P_i, T_i, F_i, \lambda)$, $i=1, 2$, with \mbox{$P_1\cap P_2=\emptyset$} and $T_1\cap T_2=\emptyset$.
    Their \emph{right-synchronized product} $N_1 \rightproduct N_2$ is the labeled Petri net
    \[
        N_1 \rightproduct N_2 = (\Sigma, P_1 \dotcup P_2, T_1 \dotcup T, F, \lambda)
        \ ,
    \]
    where for the transitions $t_1\in T_1$, $\lambda$ and $F$ remain unchanged.
    The new transitions are
    \begin{align*}
        &T = \Set{ \mathit{merge}(t_1,t_2) }{t_1 \in T_1, t_2 \in T_2, \ \lambda_1(t_1) = \lambda_2(t_2) }
        \quad \text{ with}
        \\
        &\lambda(\mathit{merge}(t_1,t_2)) = \lambda_1(t_1)=\lambda_2(t_2)\ ,
        \\
        &F(p_i,\mathit{merge}(t_1,t_2)) = F_i(p_i,t_i),
        F(\mathit{merge}(t_1,t_2),p_i) = F_i(t_i,p_i)
       \text{ for } p_i \in P_i, i = 1,2.
   \end{align*}
\end{definition}
As indicated by the name \emph{right-synchronized}, the transitions of $N_1$ can be fired without synchronization, while the transitions of $N_2$ can only be fired if a transition of $N_1$ with the same label is fired simultaneously.

Consider a Petri net $N$ with initial marking $M_0$. 
We compute the right-synchronized product $N'=N\rightproduct N_{\linof{\product}}$, take the initial marking $M_0'$ that coincides with $M_0$ but puts a token on the initial place of $N_{\linof{\product}}$, and focus on the counting places
\mbox{$X = \Set{ p_{\Gamma} }{ (\Prj{w_{\Sigma}}{\Gamma})^*\text{ is a block in $\product$} }$.}
The following correspondence holds.

\begin{lemma}
\label{RightProduct}
    $\lang{\linof{\product}} \subseteq \dc{\lang{N, M_0, M_\emptyset}}$ if and only if the places in $X$ are simultaneously unbounded in $N'$ from $M_0'$. Here $M_\emptyset$ is the zero marking, \ie $M_\emptyset (p) = 0$ for all $p$.
\end{lemma}

\begin{proof}
    Let 
    \(
        \linof{\product} = a_1 (\Prj{w_{\Sigma}}{\Sigma_1})^* a_2 \ldots a_k (\Prj{w_{\Sigma}}{\Sigma_k})^* a_{k+1}
    \)
    and let $N'=N\rightproduct N_{\linof{\product}}$ be the right-synchronized product as above.
    
    We first assume that the places in $X$ are simultaneously unbounded in $N'$.
    Given a word $w \in \lang{\linof{\product}}$, we need to find $v$ such that $w \subword v \in \lang{N, M_0, M_\emptyset}$.
    The word is of the shape
    \[
        w = a_1 (\Prj{w_{\Sigma}}{\Sigma_1})^{n_1} a_2 \ldots a_k (\Prj{w_{\Sigma}}{\Sigma_k})^{n_k} a_{k+1}
        \ .
    \]
    Define $n = \max n_i$ and let $\sigma$ be the run that creates at least $n$ tokens in each place of $p_{\Sigma_i} \in X$.
    Since the run creates at least $n$ tokens in $p_i$, it has to fire the transition leaving the block $\Prj{w_{\Sigma}}{\Sigma_i})$ $n$ times.
    This transition is a synchronized transition, \ie of type $\mathit{merge}(t,t')$.
    The fact that it could be fired means that before it, we have actually seen the synchronized transitions corresponding to the rest of the block, and before the block the synchronization transition corresponding to $a_i$.
    Altogether, we obtain that
    \[
    w' = a_1 (\Prj{w_{\Sigma}}{\Sigma_1})^{n} a_2 \ldots a_k (\Prj{w_{\Sigma}}{\Sigma_k})^{n} a_{k+1}
    \ .
    \]
    is a subword of $\lambda(\sigma)$, and by the choice of $n$, $w$ is a subword of $w'$.
    
    Towards a proof for the other direction, assume that $\lang{\linof{p}} \subseteq \dc{\lang{N', M_0, 0}}$ holds.
    Given any $n$, consider the word
    \[
        w = a_1 (\Prj{w_{\Sigma}}{\Sigma_1})^n \ldots a_k (\Prj{w_{\Sigma}}{\Sigma_k})^n \in \lang{\linof{p}}
        \ .
    \]
    Since $w \in \dc{\lang{N', M_0, 0}}$, there is a valid firing sequence $\sigma$ with $w \subword \lambda(\sigma)$.
    We consider the run $\sigma'$ of $N'$ that we construct as follows:
    For each $t$ in $\sigma$, whenever $\lambda(t)$ is present in $w$, we fire the synchronized transition $merge(t,t')$ (with suitable $t'$), whenever it is not present, we fire the non-synchronized transition $t$.
    If this run is a valid firing sequence, it is immediate that it generates $n$ tokens in each place $p_{\Sigma_i}$:
    Each block $ \Prj{w_{\Sigma}}{\Sigma_i}) $ is left $n$ times in $w$, so we trigger the synchronized transition that generates a token in $p_{\Sigma_i}$ $n$ times.
    
    We have to argue why $\sigma'$ is a valid run.
    First note that on the places of $N$, firing the non-synchronized version $t$ or firing a synchronized version $merge(t,t')$ (for arbitrary suitable $t'$) has the same effect.
    This shows that the non-synchronized transitions occurring in $\sigma'$ can be fired, and the synchronized transitions satisfy the enabledness-condition on the places of $N$, since $\sigma$ was a valid firing sequence of $N$.
    
    We still have to argue why the enabledness-condition on the places of $N_{\linof{p}}$ for the synchronized transitions is also satisfied.
    This is since $N_{\linof{p}}$ was constructed as net with language $\lang{\linof{p}}$, and we only use the synchronized transitions for the subword $w \in \lang{\linof{p}}$.
\end{proof}
The lemma does not yet involve the final marking $M_f$.
We modify $N'$ and $X$ such that simultaneous unboundedness implies $\lang{\linof{\product}} \subseteq \dc{\lang{N, M_0, M_f}}$.
The idea is to introduce a new place $p_f$ that can become unbounded only after $M_f$ has been covered.
To this end, we also add a transition $t_{f}$ that consumes $M_f(p)$ tokens from each place $p$ of $N$ and produces one token in $p_{f}$. 
We add another transition $t_{\mathit{pump}}$ that consumes one token in $p_{f}$ and creates two tokens in $p_{f}$.
Call the resulting net $N''$. 
The new initial marking $M_0''$ coincides with $M_0'$ and assigns no token to $p_{f}$.

Note that we do not enforce that $t_f$ is only fired after all the rest of the computation has taken place.
We can rearrange the transitions in any valid firing sequence of $N''$ to obtain a sequence of the shape $\sigma.{t_f}^k.{t_{\mathit{pump}}}^{k'}$, where $\sigma$ contains neither $t_f$ nor $t_{\mathit{pump}}$.

\begin{lemma}
    $\lang{\linof{\product}} \subseteq \dc{\lang{N, M_0, M_f}}$ iff the places in $X\cup\set{p_{f}}$ are simultaneously unbounded in $N''$ from $M_0''$.
\end{lemma}
To conclude the proof of Theorem~\ref{Theorem:SREIPN}, it remains to argue that the generated instance for $\mathsf{SUPPN}$ is polynomial in the input, \ie in $(N,M_0,M_f)$ and $\product$.
The expression $\linof{\product}$ is certainly linear in $\product$, and the net $N_{\linof{\product}}$ is polynomial in $\linof{\product}$. 
The blow-up caused by the right-synchronized product is at most quadratic, and adding the transitions and the places to deal with $M_f$ is polynomial.
The size of $M_0''$ is polynomial in the size of $M_0$ and $\product$.
Altogether, the size of $N''$, $X \cup \set{p_f}$, and $M_0''$ (which together form the generated instance for $\mathsf{SUPPN}$) is polynomial in the size of the original input.

%% file: content/sre_downward_bpp.tex

\subsection{BPP Nets}
\label{Subsection:SREDCBPP}

We show that the problem  of deciding whether the language of an SRE is included in the downward closure of a BPP net language is $\NPTIME$-complete.

\begin{theorem}
\label{Theorem:SREDBPP}
    \sred\ for BPP nets is $\NPTIME$-complete.
\end{theorem}
Hardness can be shown using a reduction from BPP coverability, which is $\NPTIME$-complete, similar to Lemma~\ref{Lemma:SREIPNcomplete}.
The hardness of BPP coverability itself can be easily shown by a reduction from $\mathsf{SAT}$, similar to the proof of the $\NPTIME$-hardness of reachability in BPP nets~\cite{Esparza19972}.

For membership in $\NPTIME$, we give a reduction to satisfiability of an existential formula in \emph{Presburger arithmetic}, the first-order theory of the natural numbers with addition, subtraction, and order. 

\begin{definition}
    Let $\mathcal{V}$ be a set of variables with elements $x,y$.
    The set of terms $t$ in Presburger arithmetic and the set of formulas $\varphi$ are defined as follows:
    \begin{align*}
    t &::= 0\bnf 1\bnf x\bnf t-t\bnf t+t
    &
    \varphi &::= t\leq t\bnf \neg \varphi\bnf \varphi\lor \varphi \bnf \exists x \colon \varphi
    \ .
    \end{align*}
\end{definition}
An \emph{existential} Presburger formula takes the form $\exists x_1 \ldots \exists x_n \colon \varphi$ where $\varphi$ is a quantifier-free formula.
We shall also write positive Boolean combinations of existential formulas.
By an appropriate renaming of the quantified variables, any such formula can be converted into an equivalent existential Presburger formula.
We write $\varphi(\vec x)$ to indicate that (at most) the variables $\vec x =x_1,\ldots,x_k$ occur free in $\varphi$.
Given a function $M$ from $\vec x$ to $\mathbb{N}$, the meaning of \emph{$M$ satisfies $\varphi$} is as usual and we write $M \models \varphi$ to denote this.   
We rely on the following complexity result:

\begin{theorem}[\cite{Scarpellini1983}]
\label{Theorem:SatisfiabilityEPA}
    Satisfiability in existential Presburger arithmetic is $\NPTIME$-complete.
\end{theorem}
Note that $\lang{\sre} \subseteq \dc{\lang{N, M_0, M_f}}$ iff the inclusion holds for every product $\product$ in $\sre$.
Given such a product, we construct a new BPP net $N'$ and an existential Presburger formula $\psi(P')$ such that $\lang{\product} \subseteq \dc{\lang{N, M_0, M_f}}$ iff  
there is a marking $M'$ reachable in $N'$ from a modified initial marking $M_0'$ with $M'\models \psi$. 
This concludes the proof with the help of the following characterization of reachability in BPP nets in terms of existential Presburger arithmetic.

\begin{theorem}[\cite{Verma2005,Esparza19972}]
\label{Theorem:Verma}
    Given a BPP net $N= (\Sigma,P,T,F,\lambda)$ and an initial marking $M_0$, one can compute in polynomial time an existential Presburger formula $\Psi(P)$ so that for all markings $M$: $M \models \Psi(P)$ if and only if $\TEnable{M_0}{\sigma}{M}$ for some $\sigma \in T^*$.
\end{theorem}
After constructing $N'$ and $\psi(P')$, we may use Theorem~\ref{Theorem:Verma} to construct formula $\Psi(P')$ that characterizes reachability in $N'$.
We have that $\lang{\product} \subseteq \dc{\lang{N, M_0, M_f}}$ if and only if $\psi(P') \wedge \Psi(P')$ is satisfiable, which can be checked in $\NPTIME$ using Theorem~\ref{Theorem:SatisfiabilityEPA}.

Key to the construction of $N'$ is a characterization of the computations that need to be present in the BPP net for the inclusion $\lang{\product} \subseteq \dc{\lang{N, M_0, M_f}}$ to hold.
Wlog., in the following we will assume that the product takes the shape 
\begin{align*}
(a_1+\varepsilon)\Sigma_1^*(a_2+\varepsilon) \ldots \Sigma_{n-1}^*(a_n+\varepsilon),
\end{align*} 
where $\Sigma_1,\ldots,\Sigma_{n-1} \subseteq \Sigma$ and $a_1,\ldots,a_n \in \Sigma$.
For this language to be included in $\dc{\lang{N, M_0, M_f}}$, the BPP should have a computation with parts $\sigma_i$ containing $a_i$ and parts $\rho_i$ between the $\sigma_i$ that contain all letters in $\Sigma_i$ and that can be repeated. 
To formalize the requirement, recall that we use $w_{\Sigma}$ for a total order on the alphabet and $\Prj{w_{\Sigma}}{\Sigma_i}$ for the projection to $\Sigma_i\subseteq \Sigma$. 

Moreover, we define $M \deq{c} M'$, with $c$ the constant defined in Lemma~\ref{bpppump},  if for all places $p\in P$ we have $M'(p) < c$ implies $M(p) \leq M'(p)$. 

\begin{definition}
    Let $p$ be a product. 
    The BPP net $N$ together with the markings $M_0, M_f$ admits a \emph{$p$-witness} if there is a computation
    \[
        M_0 = M_1 \fire{\sigma_0} M_1' \fire{\rho_1} M_2 \fire{\sigma_1} M_2' \fire{\rho_2} \ldots M_{n-1}' \fire{\rho_{n-1}} M_n \fire{\sigma_n} M_n' \deq{c} M_{n}'
        \ ,
    \]
    \ie there are markings markings $M_1,M'_1,\ldots, M_{n}, M_{n}'$ and firing sequences $\sigma_i$, $\rho_i$ that satisfy $\TEnable{M_{i}}{\sigma_i}{M'_{i}}$ for all $i \in \oneto{n}$, $\TEnable{M'_{i}}{\rho_i}{M_{i+1}}$ for all $i \in  \oneto{n-1}$, and moreover:\\
    (1) $a_{i} \subword \lambda(\sigma_i)$, for all $i \in  \oneto{n}$,\\
    (2) $\Prj{w_\Sigma}{\Sigma_i} \subword \lambda(\rho_i)$ for all $i \in \oneto{n-1}$,\\
    (3) $M_i'\deq{c}M_{i+1}$ for all $i \in \oneto{n-1}$, and\\
    (4) $M_1 = M_0$ and $M_{f} \deq{c} M_{n}'$.
\end{definition}
In a $p$-witness, (1) enforces that the $a_i$ occur in the desired order, and the first part of (2) requires that $\Prj{w_\Sigma}{\Sigma_i}$ occurs in between.
Property~(3) means that each $\rho_i$ (and thus $\Prj{w_\Sigma}{\Sigma_i}$) can be repeated.
Property~(4) enforces that the computation still starts in the initial marking and can be extended to cover the final marking.

The following proposition reduces the problem \sred\ for BPP nets to checking whether the BPP admits a $p$-witness.
\begin{proposition}
\label{markingrelation}
    $\lang{\product} \subseteq \dc{\lang{N, M_0,M_f}}$ holds iff  $(N, M_0,M_f)$ admits a $p$-witness.
\end{proposition}

\begin{proof}
    Recall that we consider a product of the form
    \[
        \product = (a_1+\varepsilon) \Sigma_1^* (a_2+\varepsilon) \Sigma_2^* \cdots \Sigma_{n-1}^* (a_n+\varepsilon) \ .
    \]
    Assume that $(N,M_0,M_f)$ admits a $p$-witness $(M_1,M'_1, \cdots,M_n,M'_n)$ such that
    \[
       M_0 = M_1 \fire{\sigma_0} M_1' \fire{\rho_1} M_2 \fire{\sigma_1} M_2' \fire{\rho_2} \ldots M_{n-1}' \fire{\rho_{n-1}} M_n \fire{\sigma_n} M_n' \deq{c} M_{n}'
       \ ,
    \]
    satisfying the required properties.
    We will show $\lang{p} \subseteq \dc{\lang{N,M_0,M_f}}$ by proving that for any word $w \in \lang{p}$, there is a run $\TEnable{M_0}{\sigma}{M''}$ such that $w \subword \lambda(\sigma)$ and $M'_n \deq{c} M''$ (and hence also $M_f \leq M''$).
    Let $w = x_1v_1\cdots v_{n-1}x_n$, where $x_i \in \{a_i , \varepsilon \}$ and $ v_i \in \Sigma_i^*$. 
    
    In sequel, we will prove that for every prefix $w' = x_1v_1\cdots x_i$, we have a run of the form $\TEnable{M_0}{\gamma}{M_i'''}$ such that $w' \subword \lambda(\gamma)$ and $M_i' \deq{c} M_i'''$.
    Similarly, we show that for any prefix of the form $w' = x_1v_1\cdots x_iv'_i$, where $v'_i$ is a prefix of $v_i$, we have a run of the form $\TEnable{M_0}{\gamma}{M_i''}$ such that $w' \subword \lambda(\gamma)$ and $M_i \deq{c} M_i''$. 
    
    We prove both statements simultaneously by induction.
    Consider the first statement for $i = 1$:
    We have $w' = x_1$, and $\TEnable{M_1}{\sigma_1}{M'_1} $ is the required run by Property~(1).

    Consider the second statement for some $i$ for which we assume that the first statement holds.
    Let
    \[
        w'' =  x_1v_1 \ldots x_{i}v'_{i}
        \ ,
    \]
    where $v'_{i}.a$ is a prefix of $v_{i}$.
    To show the required statement, we consider an inner induction on the length of $v'_{i}$.
    In the base case, $v'_i = \varepsilon$, and the statement is immediate by the hypothesis of the outer induction.
    Assume that $v'_{i} = v''_{i}.b$, then by the inner induction, we have a run $\TEnable{M_1}{\gamma}{M_{i}''}$ such that $x_1v_1 \cdots x_{i}v''_{i} \subword \lambda(\gamma)$ and $M_{i} \deq{c}M''_{i}$.
    If for all places $p \in P$, $M''_{i}(p) < c$ holds, then we have $M_{i} \leq M''_{i}$.
    We prolong $\gamma$ by $\rho_{i}$, which is possible by Property~(3), and get the required run by Property~(2).
    Suppose the set of places $X$ to which more than $c$ tokens are assigned is non-empty.
    Using Lemma~\ref{bpppump} repeatedly for each place in $X$, we pump the run to create enough tokens to be able to execute the rest of the run.
    We obtain a run $\TEnable{M_1}{\gamma'}{M_{i}^{*}}$ such that $\lambda(\gamma) \subword \lambda(\gamma')$, for all $x \in X$ $M^{*}_{i}(x) - M''_{i}(x) \geq \card{\rho_i}$ and $M''_{i} \deq{c} M^{*}_{i}$.
    This is can be prolonged by $\rho_{i}$ to obtain the desired computations.
    This concludes the inner induction and the proof of the second statement for $i$.
    
    It remains to prove the first statement for $i+1$, assuming that the second statement holds for $i$.
    Consider 
    \[
        w' = x_1 v_1 x_2\ldots x_i v_i x_{i+1}
        \ .
    \]
    By induction, we get a run $\TEnable{M_1}{\gamma}{M_{i}''}$ such that $x_1v_1x_2 \cdots x_i v_{i} \subword \lambda(\gamma)$ and $M_{i} \deq{c}M''_{i}$. 
    Suppose for all places $p \in P$, $M''_{i} < c$, then we have $M_{i} \leq M''_{i}$.
    Hence we can easily extend the computation by $\TEnable{M''_{i} }{\sigma_i}{M'''_{i}}$, which gives us the required run.
    Otherwise, we proceed as for the second statement and pump up the values in these places to be greater than the size of $\sigma_{i}$.
    Afterwards, we can extend the computation by $\sigma_{i}$, obtaining the desired run.
    
    For the other direction, consider the word
    \[
        w = a_1.(\Prj{w_\Sigma}{\Sigma_1})^{\ell \cdot c +1}.a_2.(\Prj{w_\Sigma}{\Sigma_2})^{\ell \cdot c +1}.a_3 \cdots a_n \in L(p)
        \ .
    \]
    Since we have $\lang{p} \in \dc{\lang{N,M_0,M_f}}\downarrow$, we have a run of the form
    \[
    M_1\move{ \alpha_1 } J'_1 \move{ \beta_1 } J_1 \move{ \alpha_2 } J'_2  \move{ \beta_2 } J_2 \cdots J'_n
    \ ,
    \]
    such that $a_i \subword \lambda(\alpha_i)$ and $ \Prj{w_\Sigma}{\Sigma_i} \subword \lambda(\beta_i)$.
    Since the length of $\beta_i$ is $\ell \cdot c + 1$, there have to be markings between $J'_{i}$ and $J_i$ such that
    \[
    J'_{i}\move{ \beta^1_i } J^1_i \move{ \beta^2_i } J^2_i \move{ \beta^3_i } J_i
    \ ,
    \]
    where $J^1_i \deq{c} J^2_i$.
    Now, we let $M'_i = J^1_{i}$, $M_i = J^2_{i}$, $\sigma_1 = \alpha_1.\beta_1^1$, $\sigma_i = \beta_{i-1}^3.\alpha_i.\beta_i^1$ and $\rho_i = \beta_i^2$.
    This gives us the required $p$-witness.
\end{proof}
We now reduce the problem of finding a $p$-witness to finding in a modified unlabeled BPP net $N' = (\emptyset, P',T',F',\lambda')$ a reachable marking that satisfies a Presburger formula $\Psi^{M_0}_{M_f} (P')$.  
The task is to identify $2n$ markings that are related by $2n-1$ computations as required by a $p$-witness.  
The idea is to create $2n-1$ replicas of the BPP net and run them independently to guess the corresponding computations $\sigma_i$ \resp $\rho_i$. 
The Presburger formula $\Psi^{M_0}_{M_f}$ will check that the target marking reached with $\sigma_i$ coincides with the initial marking for $\rho_{i}$, and the target marking reached with $\rho_i$ is the initial marking of $\sigma_{i+1}$. 
To this end, the net $N'$ remembers the initial marking that each replica started from in a full copy (per replica) of the set of places of the BPP net.
Furthermore $\Psi^{M_0}_{M_f}$ checks that each $\rho^i$ can be repeated by ensuring that the final marking in the corresponding replica is larger than the initial marking.
As initial marking for $N'$, we consider the marking $M_\emptyset$ with $M_\emptyset (p) = 0$ for all $p$.

Formally, the places of $N'$ are
\[
    P'  = \bigcup _{i \in \oneto{2n-1}} B_i \cup E_i \cup L_i
    \ .
\]
Here, $E_i = \Set{ e_i^p}{p \in P}$ and $B_i = \Set{ b_i^p}{p \in P }$ are $2n-1$ copies of the places of the given BPP net.  
The computation $\sigma_i$ or $\rho_i$ is executed on the places $E_i$, which will hold the target marking $M_i'$ or $M_{i+1}$ reached after the execution. 
The places $B_i$ remember the initial marking of the replica and there are no transitions that take tokens from them.
The places $L_i$ record occurrences of $a_i$ and of symbols from $\Sigma_i$, depending on whether $i$ is odd or even. 
For all $i\in \oneto{n}$, we have $L_{2i-1} = \{l_{2i-1}\}$.
For all $i\in \oneto{n-1}$, we set $L_{2i} = \Set{l_{2i}^a}{a \in \Sigma_{i}}$ otherwise. 
The transitions are
\[
    T' =  \bigcup _{i \in [1..2n-1]} \TC_i \cup \TE_i
    \ .
\]
The $\TC_i = \Set{ \text{tc}_i^p}{ p \in P }$ populate $E_i$ and $B_i$.  
The transitions in $\TE_i =  \Set{ \text{te}_i^t}{t \in T }$ together with $E_i$ form a replica of the BPP net. 
The flow relation $F'$ is defined as follows, where numbers omitted are zero:
\begin{enumerate}[(1)]
    \item For all $i \in \oneto{2n-1}$, $p \in P$, $F'(\text{tc}_i^p, e_i^p) = F'(\text{tc}_i^p,b_i^p) = 1$.
    \item For all $i \in \oneto{2n-1}$, $p \in P$, $t \in T$, $F'( \text{te}_i^t, e_i^p) = F(t,p)$ and $F'( e_i^p,\text{te}_i^t)=F(p,t)$.
    \item For all $i \in \oneto{n}$, $t \in T$ with $\lambda(t) = a_{i}$, $F'(\text{te}_{2i-1}^t, l_{2i-1})=1$.
    \item For all $i \in \oneto{n}$, $t \in T$ with $\lambda(t) =a\in \Sigma_{i}$, $F'(\text{te}_{2i}^t, l^{a}_{2i}) = 1$
\end{enumerate}

The Presburger formula \wrt the initial and final markings $M_0$ and $M_f$ of $N$ has the places in $P'$ as variables.
It takes the shape
\begin{align*}
    \Psi^{M_0}_{M_f}(P') = &\Psi_1(P')\wedge \Psi_2(P')\wedge \Psi_3(P') \wedge \Psi_4(P')\wedge \Psi_5^{M_0}(P')\wedge \Psi_6^{M_f}(P')
    \ ,
\end{align*}
where
\begin{align*}
    \Psi_1(P')
    &=
    \bigwedge_{i \in \oneto{2n-2}} \bigwedge_{p \in P}
    e_i^p = b_{i+1}^p
    \\
    \Psi_2(P')
    &=
    \bigwedge_{i \in \oneto{n-1}} \bigwedge_{p \in P} 
    (e_{2i}^p < c \ \rightarrow \ b_{2i}^p \leq e_{2i}^p)
    \\
    \Psi_3(P')
    &=
    \bigwedge_{i\in \oneto{n}}
    l_{2i-1}>0
    \\
    \Psi_4(P')
    &=
    \bigwedge_{i\in \oneto{n-1}} \bigwedge_{a \in \Sigma_{i}}
    l_{2i}^a>0
    \\
    \Psi_5^{M_0}(P')
    &=
    \bigwedge_{p\in P}
    b_1^p = M_0(p)
    \\
    \Psi_6^{M_f}(P')
    &=
    \bigwedge_{p\in P}
    (e_{2n-1}^p < c \ \rightarrow \ e_{2n-1}^p \geq M_f(p))
    \ .
\end{align*}
Formula $\Psi_1$ states that $\sigma_i$ ends in the marking $M_{i}'$ that $\rho_i$ started from, and similarly $\rho_i$ ends in $M_{i+1}$ that $\sigma_{i+1}$ started from.
Formula $\Psi_2$ states the required $\deq{c}$ relation.
To make sure we found letter $a_i$, we use $\Psi_3$.
With $\Psi_4$, we express that all letters from $\Sigma_i$ have been seen.
Conjunct $\Psi_5^I$ says that the places $b_1^p$ have been initialized to the value given by the initial marking $I$.
Formula $\Psi_6^F$ states the condition on covering the final marking. 
The correctness of the construction is the next lemma. 
Note that the transitions $\TC_i$ are always enabled. 
Therefore, we can start in $N'$ from the initial marking $M_{\emptyset}$ that assigns zero to every place. 

\begin{proposition}
    \label{CorrectnessPresburger}
    There are $\sigma'$ and $M'$ so that $\TEnable{M_{\emptyset}}{\sigma'}{M'}$ in $N'$ and $M' \models \Psi^{M_0}_{M_f}$
    if and only if $(N,M_0,M_f)$ admits a $p$-witness.
\end{proposition}

\begin{proof}
    Assume a $p$-witness $M_1, \ldots, M_{2n}$.
    We construct a run $M_\emptyset \move{\sigma'} M'$ of $N'$ with \mbox{$M' \models \Psi^{M_0}_{M_f}$} as follows:
    \[
    \sigma' = \gamma_1 \alpha_1 \gamma_1' \beta_1 \gamma_2 \alpha_2 \gamma_2' \beta_2 \ldots \gamma_n \alpha_n \gamma_n' \beta_n
    \ ,
    \]
    where
    the $\alpha_i$ corresponds to $\sigma_i$ executed on the places $E_{2i-1}$ by using the transitions in $\text{TE}_{2i-1}$ (\ie if a transition $t \in T$ is used in $\sigma_i$, then the transition $\text{te}_{2i-1}^t \in \text{TE}_{2i-1}$ is used in $\alpha_i$).
    Similarly, the $\beta_i$ correspond to the $\rho_i$ executed on $E_{2n}$ using transitions in $\text{TE}_{2n}$.
    The $\gamma_i$ and $\gamma_i'$ populate the set of places $E_i$ accordingly:
    $\gamma_1$ produces $M_1(p)$ many tokens on each place $e_1^p$ of $E_1$ using the transition in $\text{TC}_1$.
    For $i > 1$, $\gamma_i$ produces $M_{2i-1}(p)$ many tokens on each place $e_{2i-1}^p$ of $E_{2i-1}$, and $\gamma_i'$ produces $M_{2i}(p)$ many tokens on each place $e_{2i}^p$ of $E_{2i}$.
    As a by-product, the places in each $B_i$ are also populated.
    It is easy to check that the marking $M'$ with $M_\emptyset \move{\sigma'} M'$ indeed satisfies $\Psi^{M_0}_{M_f}$.
    
    For the other direction, assume that a computation $M_\emptyset \move{\sigma'} M'$ with $M' \models \Psi^{M_0}_{M_f}$ is given.
    First, observe that the transitions in $\TC_i$ and $\TE_i$ are not dependent on each other and hence can be independently fired.
    Furthermore, for $i \neq j$, the transitions in $\TE_i$ and $\TE_j$ are independent of each other.
    Therefore, we may assume that in $\sigma'$, the copy transitions in $\TC = \bigcup_{i \in \oneto{2n-1}}\TC_i$ are fired first, then the transition in $\TE_1$ followed by transition in $\TE_2$ and so on.
    All together, we may assume that the computation is of the form
    \[
    \TEnable{M_\emptyset}{\sigma''.\sigma'_1.\cdots.\sigma'_{2n-1}}{M'}
    \ ,
    \]
    where $\sigma'' = \Prj{\sigma'}{\TC}$ and $\sigma'_i = \Prj{\sigma'}{\TE_i}$ for all $i \in \oneto{2n-1}$
    Note that for each $i \in \oneto{2n-1}$, the transition sequence $\sigma_i'$ in $N'$ induces a transition sequence $\alpha_i$ in the original net $N$ by using transition $t \in T$ instead of $\text{te}_i^t \in \text{TE}_i$.
    
    The initial phase $\sigma''$ populates each $E_i$ with some initial marking.
    This initial marking is also copied to the places $B_i$, and these places are not touched during the rest of the computation.
    We may obtain a marking $J_i$ of $N$ for each $i \in \oneto{2n-1}$ by $J_i (p) = M'(b_i^p)$.
    For each $i \in \oneto{2n-1}$, we obtain a marking $K_i$ of $N$ by considering the assignment of tokens to the places of $E_i$ by $M'$, \ie $K_i (p) = M' (e_i^p)$.    
    
    We claim that $M_0, K_1, K_2, \ldots K_{2n-1}$ is the required $p$-witness.
    To argue that they indeed satisfy the Properties (1) to (4), we use the fact that $J_i \move{\sigma_i'} K_i$ as well as $M' \models \Psi^{M_0}_{M_f}$.
    \begin{enumerate}[(a)]
        \item
        Since $M' \models \Psi_5^{M_0}$, we have $J_1 = M_0$.
        \item
        Since $M' \models \Psi_6^{M_f}$, we have $K_{2n-1} \deq{c} M_f$.
        \item
        Since $M' \models \Psi_1$, we have $J_{i+1} = K_i$.
        \item
        Since $M' \models \Psi_2$, we have $J_{2i} \deq{c} K_{2i}$ for all $i \in [1..n[$.
        \item
        Since $M' \models \Psi_3$, we have for all $i \in [1..n]$, $a_i \subword \lambda( \sigma'_{2i-1} )$.
        \item
        Since $M' \models \Psi_4$, we have for all $i \in [1..n[$, for all $a \in \Sigma_i$, $a \subword \lambda (\sigma_{2i} )$.     
    \end{enumerate}
    
    \noindent
    We conclude that Property (4) holds using (a) and (b).
    We conclude Property~(2), $a_i \subword \lambda( \sigma_i )$, using (e) and Property~(3), $\Prj{w_\Sigma}{\Sigma_i} \subword \lambda (\rho_i)$, using (f).
    Finally, (b) and (d) yields the required Property~(3).
\end{proof}

%% file: content/sre_upward.tex
\section{SRE Inclusion in Upward Closure}
\label{Section:SREU}

Rather than computing the upward closure of a Petri net language we now check whether a given SRE under-approximates it. 
Formally, the problem is defined as follows.

\begin{problem}
	\problemtitle{SRE Inclusion in Upward Closure (\sreu)}
	\probleminput{A Petri net instance $(N,M_0,M_f)$, an SRE $sre$.}
	\problemquestion{$\lang{sre} \subseteq \uc{\lang{N,M_0,M_f}} $?}	
\end{problem}

\subsection{Petri Nets}

\begin{theorem}
\sreu\ is $\EXPSPACE$-complete for Petri nets. 
\end{theorem}
The $\EXPSPACE$\ lower bound is immediate by hardness of coverability for Petri nets and can be proven similar to Lemma~\ref{Lemma:SREIPNcomplete}.
The upper bound is due to the following fact: We only need to check whether the set of minimal words in the language of the given SRE is included in the upward closure of the Petri net language.
Note that the minimal word of a product can be computed as follows:
\begin{align*}
    \text{min} (a) &= a
    &
    \text{min} (p.p') &= \text{min} (p) . \text{min} (p')
    \\
    \text{min} (a + \varepsilon) &= \varepsilon
    &
    \text{min} (\Gamma^*) &= \varepsilon\ .
\end{align*}
For an SRE $sre = p_1 + \ldots + p_n$, we have that the set of minimal words is a subset of $\set{ \text{min } p_1, \ldots, \text{min } p_n }$.
We have that
\[
    \lang{sre} \subseteq \uc{\lang{N,M_0,M_f}}
    \quad
    \text{ iff }
    \quad 
    \text{min } p_i \in  \uc{\lang{N,M_0,M_f}} \text{ for all } i \in \oneto{n}
    \ .
\]
For each product, the membership check $\text{min } p_i \in  \uc{\lang{N,M_0,M_f}}$ can be reduced in polynomial time to coverability in Petri nets.
Since the number of minimal words in the SRE language is less than the size of the SRE, and coverability is well-known to be in $\EXPSPACE$\ \cite{Rackoff78}, we obtain our $\EXPSPACE$\ upper bound.

\subsection{BPP Nets}
\begin{theorem}
\label{sreuNPcomplete}
    \sreu\ is $\NPTIME$-complete for BPP nets.
\end{theorem}
As before, the hardness is by a reduction of the coverability problem for BPP nets. 
For the upper bound, the algorithm is similar to the one for checking the inclusion of an SRE in the downward closure of a BPP language.  
\begin{proof}
    To check $\lang{sre} \subseteq \uc{\lang{N,M_0,M_f}}$, it is sufficient to check $\lang{\product} \subseteq \uc{\lang{N,M_0,M_f}}$ for each product in $sre$.
    Consider one such product $\product$.
    The inclusion $\lang{\product} \subseteq \uc{\lang{N,M_0,M_f}}$ holds iff the minimal word of $\lang{\product}$, say $\min p = a_1 \ldots a_n$, belongs to $\uc{\lang{N,M_0,M_f}}$. 
    This in turn holds iff one of its subwords is in $\lang{N,M_0,M_f}$.
    We check this by deciding whether a reachable marking $M$ in a different net $N'$ satisfies a Presburger formula $\Psi$. 
    
    We describe the BPP net $N'$ and the Presburger formula $\Psi$ that together characterize the subwords of $\min(\product)$ included in the language of the BPP net. 
    Net $N'$ is constructed similar to the net $N'$ from Section~\ref{Subsection:SREDCBPP}.
    We have two copies of the places for each $i \in \oneto{n}$, the places $B_i$ hold a copy of the guessed marking and $E_i$ provides a copy $e_i^p$ of the BPP net places $p$. 
    Additionally for each $i$ we have a place $L_i = \{l_i\}$.
    The transitions $\TC_i$ populate the copy $E_i$ of the BPP net and store the same marking in $B_i$.
    The transitions $\TE_i$ contain a copy $\text{te}_i^t$ of each BPP net transition $t$. 
    To check for a subword of $a_1\ldots a_n$, in each stage $i$ we only enable transitions $t$ that are either labeled by $\varepsilon$ or $a_i$, \ie for all $p \in P$, $t \in T$, if $\lambda(t) = \varepsilon$ or $\lambda(t) = a_i$ we have $F'( \text{te}_i^t, e_i^p)=F(t,p)$ and $F'( e_i^p,\text{te}_i^t)=F(p,t)$. 
    We also count the number of times a transition labeled $a_i$ is executed using place $l_i$, \ie for all $t \in T$ such that $\lambda(t) = a_i$, we let $F'(\text{te}_i^t,l_i) = 1$. 
    
    Now the required Presburger formula $\Psi$ --- apart from checking that
    (1)~the net starts with the initial marking,
    (2)~covers the final marking,
    (3)~the guessed marking in each stage is the same as the marking reached in the previous stage --- 
    also checks whether in each stage at most one non-epsilon transition is used, $\bigwedge_{i \in 1..n} l_i \leq 1$. 
    This guarantees we have seen a subword of $a_1\ldots a_n$. 
    The initial marking $M_{\emptyset}$ is one that assigns zero to all places.
    It is easy to see that $\lang{p} \subseteq \uc{\lang{N,M_0,M_f}}$ iff there is a computation $M_{\emptyset} \move{\sigma} M$ in $N'$ such that $M \models \Psi$.
\end{proof}

%% file: content/being_ucdc.tex

\section{Being Upward/Downward Closed}
\label{Section:BeingUCDC}

We now study the problem of deciding whether a Petri net language actually is upward or downward closed, \ie whether the closure that we can compute is actually a precise representation of the system's behavior. 
Formally, the problems \buc\ and \bdc\ are defined as follows.n 
\begin{problem}
    \problemtitle{Being upward closed (\buc)}
    \probleminput{A Petri net instance $(N,M_0,M_f)$.}
    \problemquestion{$\lang{N,M_0,M_f} = \uc{\lang{N,M_0,M_f}}$? }
\end{problem}
\begin{problem}
    \problemtitle{Being downward closed (\bdc)}
    \probleminput{A Petri net instance $(N,M_0,M_f)$.}
    \problemquestion{$\lang{N,M_0,M_f} = \dc{\lang{N,M_0,M_f}}$? }
\end{problem}

\begin{theorem}
\label{Theorem:BUCBDCDecidable}
    \buc\ and \bdc\ are decidable for Petri nets.
\end{theorem}
Note that \mbox{$\lang{N,M_0,M_f} \subseteq \uc{\lang{N,M_0,M_f}}$} and \mbox{$\lang{N,M_0,M_f} \subseteq \dc{\lang{N,M_0,M_f}}$} trivially hold.
In both cases, it remains to decide the converse inclusion.
Now note that $\uc{\lang{N,M_0,M_f}}$ (\resp $\dc{\lang{N,M_0,M_f}}$) is a regular language for which we can construct a generating FSA by Theorem~\ref{Theorem:UCCompPN} (\resp using~\cite{HMW10}).

To prove Theorem~\ref{Theorem:BUCBDCDecidable} it is thus sufficient to show how to decide $\lang{A} \subseteq \lang{N,M_0,M_f}$ for any given FSA $A$.
This regular inclusion should be a problem of independent interest.

\begin{problem}
    \problemtitle{Containing a regular language}
    \probleminput{A Petri net instance $(N,M_0,M_f)$, FSA $A$.}
    \problemquestion{$\lang{A} \subseteq \lang{N,M_0,M_f}$? }
\end{problem}

\begin{theorem}
\label{Theorem:RegInclDecidable}
    $\lang{A} \subseteq \lang{N,M_0,M_f}$ is decidable.
\end{theorem}

To prove this theorem, we rely on a result of Esparza et.\ al~\cite{Esparza99} that involves the \emph{traces} of an FSA (\resp Petri net), labelings of computations that start from the initial state (\resp initial marking), regardless of whether they end in a final state (\resp covering marking).
For a finite automaton $A$, we define 
\[
    \Traces{A} = \big\{ w \in \Sigma^* \mid q_{\mathit{init}} \tow{w} q \text{ for some } q \in Q \big \}.
\]
Similarly, for a Petri net, we define
\[
    \Traces{N,M_0} = \Set{w \in \Sigma^*}{
                \exists \, \sigma \in T^* \colon \lambda(\sigma) = w, M_0 \move{\sigma} M\ 
                \text{ for some marking } M
        }
    \ .
\]
Note that both languages are necessarily \emph{prefix closed}, \eg if $w \in \lang{A}$ for some FSA $A$, then for any prefix $v$ of $v$, we have $v \in \lang{A}$.

\begin{theorem}[\cite{Esparza99}]
\label{Theorem:EsparzaTraces}
    The inclusion $\Traces{A} \subseteq \Traces{N,M_0}$ is decidable.
\end{theorem}
The algorithm constructs a computation tree of $A$ and $N$. 
This tree determinizes $N$ in that it tracks sets of incomparable markings reachable with the same trace. 
The construction terminates if either the set of markings becomes empty and the inclusion fails or (the automaton deadlocks or) we find a set of markings that covers a predecessor and the inclusion holds. 
The latter is guaranteed to happen due to the well-quasi ordering (wqo) of sets of markings. 
This dependence on wqos does not allow us to derive a complexity result. 

We now show how to reduce checking the inclusion $\lang{A} \subseteq \lang{N,M_0,M_f}$ to deciding an inclusion among trace languages.
Theorem~\ref{Theorem:EsparzaTraces} can be used to decide this inclusion. 
Let $(N,M_0,M_f)$ be the Petri net instance of interest, and let $A$ be the given FSA.
As language $\lang{N,M_0,M_f}$ is not prefix-closed in general, we consider the zero marking $M_\emptyset$ as the new final marking.
This yields a prefix-closed language with $\Traces{N,M_0} = L(N,M_0,M_\emptyset)$, since now all valid firing sequences give a word in the language, and prefixes of valid firing sequences are again valid firing sequences.
We still need to take the original final marking $M_f$ into account.
To do so, we modify the net by adding a new transition that can only be fired after $M_f$ has been covered.
Let $a \not \in \Sigma$ be a fresh letter.
Let $N.a$ be the Petri net that is obtained from $N$ and the given final marking $M_f$ by adding a new transition $t_{\mathit{final}}$ that consumes $M_f(p)$ many tokens from every place $p$ of $N$ and that is labeled by $a$. 
For the automaton, we use a similar trick.
Let $A.a$ be an automaton for $\lang{A}.a$ that is reduced in the sense that the unique final state is reachable from every state.

\begin{lemma}
\label{ReduceToPrefixInclusion}
    $
        \lang{A} \ \subseteq \lang{N,M_0,M_f} 
    $ 
    holds iff
    $
        \Traces{A.a} \subseteq \Traces{N.a, M_0}
    $
    holds.
\end{lemma}
\begin{proof}
    Assume the first inclusion holds and consider a word $v$ from $\Traces{A.a}$.
    We have to show membership of $v$ in $\Traces{N.a, M_0}$.
    As the unique final state of $A.a$ is reachable from every state, $v$ is a prefix of some word in the language $\lang{A}.a$, say $w.a$, where $w$ stems from $\lang{A}$. 
    The assumed first inclusion now yields $w\in \lang{N,M_0,M_f}$.
    Thus, there is a $w$-labeled computation $M_0 \fire{\sigma} M$ of $N$ with $M \geq M_f$.
    We obtain that $M_0 \fire{\sigma.t_{\mathit{final}}} M'$ is a valid computation of $N.a$, thus, $w.a = \lambda(w.t_{\mathit{final}}) \in \Traces{N.a, M_0}$.
    Since trace languages are prefix closed and $v$ is a prefix of $w.a$, we obtain $v \in \Traces{N.a, M_0}$ as desired.
    
    Assume the second inclusion holds and consider a word $w$ from $\lang{A}$. 
    The task is to prove membership of $w$ in $\lang{N,M_0,M_f}$.
    To do so, note that $w.a \in \lang{A}.a \subseteq \Traces{A.a}$.
    By the assumption, we have $w.a \in \Traces{N.a, M_0}$.
    Thus, there is a valid computation $M_0 \fire{\sigma.t_{\mathit{final}}} M'$ of $N.a$ with $\lambda(\sigma.t_{\mathit{final}}) = w.a$.
    (Here, we have used that $t_{\mathit{final}}$ is the only $a$-labeled transition).
    Since $w \in \lang{A} \subseteq \Sigma^*$, and $a \not\in \Sigma$, we have that $\sigma$ does not contain an occurrence of $t_{\mathit{final}}$, so $M_0 \fire{\sigma} M$ is a valid computation of $N$.
    As $t_{\mathit{final}}$ could be fired in $M$, we have $M \geq M_f$ and $\sigma$ is indeed a covering computation in $N$.
    We conclude $\lambda(\sigma) = w \in \lang{N,M_0,M_f}$ as desired.
\end{proof}
Combining Lemma~\ref{ReduceToPrefixInclusion} and Theorem~\ref{Theorem:EsparzaTraces} yields the proof of Theorem~\ref{Theorem:RegInclDecidable}, which in turn proves Theorem~\ref{Theorem:BUCBDCDecidable}:
Given an FSA $A$ an a Petri net instance $(N,M_0,M_f)$ for which we should decide $\lang{A} \ \subseteq \lang{N,M_0,M_f}$, we construct $N.a$ and $A.a$, and apply Theorem~\ref{Theorem:EsparzaTraces} to decide $\Traces{A.a} \subseteq \Traces{N.a, M_0}$, which is equivalent to deciding $\lang{A} \ \subseteq \lang{N,M_0,M_f}$ by Lemma~\ref{ReduceToPrefixInclusion}.

%% file: content/conclusion.tex

\section{Conclusion}
We considered the class of Petri net languages with coverability as the acceptance condition and studied the problem of computing representations for the upward and downward closure. 
For the upward closure of a Petri net language, we showed how to effectively obtain an optimal finite state representation of size at-most doubly exponential in the size of the input. 
In the case of downward closures, we showed an instance for which the minimum size of any finite state representation is at-least non-primitive recursive. 

To tame the complexity, we considered two variants of the closure computation problem.
The first restricts the input to BPP nets, which can be understood as compositions of unboundedly many finite automata.
For BPPs, we showed how to effectively obtain an optimal finite state representation of size at-most exponential in the size of input, for both the upward and the downward closure of the language. 

The second variant takes as input a simple regular expression~(SRE), which is meant to under-approximate the upward or downward closure of a given language. 
For Petri net languages, we found an optimal algorithm that uses at-most exponential space to check whether a given SRE is included in the upward/downward closure.  
In the case of BPP nets, we showed that this problem is $\NPTIME$-complete. 

Finally, we showed that, given a Petri net, deciding whether its language actually is upward or downward closed is decidable. 
If the check is successful, the finite state descriptions we compute are precise representations of the system behavior.

An interesting problem for future work is the complexity of checking separability by piecewise-testable languages (PTL) and the size of separators.
A PTL is a Boolean combination of upward closures of single words. 
PTL-separability takes as input two languages $\calL_1$ and $\calL_2$ and asks whether there is a PTL $\calS$, called the separator, that includes $\calL_1$ and has an empty intersection with $\calL_1$. 
Taking a verification perspective, the separator is an over-approximation of the system behavior $\calL_1$ that is safe \wrt the bad behaviors in $\calL_1$. 
For deterministic finite state automata, PTL-separability was shown to be decidable in polynomial time by Almeida and Zeitoun~\cite{AlmeidaZ97}, a result that was generalized to non-deterministic automata in \cite{CzerwinskiMM13}.
Recently~\cite{CzerwinskiMartensRooijenZeitounZetzsche2015a}, Czerwi\'nski, Martens, van Rooijen, Zeitoun, and Zetzsche have show that, for full trios, computing downward closures and deciding PTL-separability are recursively equivalent. 
A full trio is a class of languages that is closed under homomorphisms, inverse homomorphisms, and regular intersection. 
Petri net languages with coverability or reachability as the acceptance condition satisfy these requirements. 
Hence, we know that PTL-separability is decidable for them~\cite{HMW10}.
The aforementioned problems, however, remain open.